\documentclass{mjsartcl}
\usepackage{header}

\title{Blindfolded pursuit with delays of your choice}
\author{
    \person{Torben Sch{\"u}renberg}{0009-0006-5947-0172}{1}%
    \and
    \person{Maximilian J. Stahlberg}{0000-0002-0190-2693}{1}%
}
\affiliation{
    \institute{1}{University of Bremen}{Germany}{\twomailnames{torsch}{maximilian.stahlberg}}{uni-bremen.de}
}
\hypersetup{
    pdfkeywords={%
        pursuit--evasion,
        hunter and rabbit,
        firefighting on graphs,
        network design}
}

\begin{document}

\maketitle

\begin{abstract}
    \noindent
    We study pursuit-evasion games on graphs with a single pursuer and an invisible evader.
    The pursuer may assign integer travel times to the edges of the graph and specify a finite sequence of vertices to query, one per time step.
    The evader then chooses a walk over the same time horizon, aiming to elude all queries.
    With unit travel times, the setting in which the evader must move at every time step is known as the hunter and rabbit game, while the variant in which the evader can wait at a vertex can be phrased as a firefighting game: the vertices of a burning graph must be extinguished, and any vertex left burning reignites its neighbors.
    For both settings, we show that the power to choose travel times allows a single pursuer to succeed in polynomial time on any graph.
    This contrasts with unweighted graphs, where the number of hunters or firefighters needed can grow linearly in the number of vertices.
    If the evader, in addition to waiting, may start at an earlier time unknown to the pursuer, we show that the pursuer still wins on any graph given exponential time.
\end{abstract}

\section{Introduction}%
\label{sec:intro}

In the classic \emph{hunter and rabbit game}~\cite{Britnell2013,Haslegrave2014,Fedorov2011}, a pursuer tries to catch an evader performing a walk on a finite graph.
The pursuer cannot see the evader but is allowed to inspect one vertex per time step: if the evader happens to be there, they are caught.
The evader, however, knows the pursuer's strategy and attempts to find a walk that eludes all observations.
The \emph{hunting number}~\cite{Abramovskaya2016} of a graph is defined as the minimum number of pursuers that have a joint winning strategy by querying one vertex per time step each.
The hunting number for general graphs of order~\(n\) can be as large as \(n - 1\) (attained by the complete graph), and pinpointing it for restricted graph classes has been an active area of research~\cite{Abramovskaya2016,Bolkema2019,Gruslys2015,BenAmeur2026}.

In this work we ask: what if the pursuer has to search a general graph alone, but can alter the playing field to their advantage?
One possible modification would be to remove vertices or edges from the graph to make it easier to search.
Of course, by removing all edges, the pursuer has the evader firmly trapped, but this is cruel and not interesting.
One may thus ask for the minimum number of vertices or edges that must be deleted to guarantee the pursuer a winning strategy.
But since the connected graphs winnable by a single pursuer form a subclass of trees~\cite{Britnell2013,Haslegrave2014}, the resulting network can remain at best minimally connected---still a major intervention if the graph is initially dense.
Here, we do not permit the pursuer to alter the graph's topology.
Instead, we only allow them to assign finite \emph{delays} to its edges.
This has the effect that the evader, when entering an edge with delay~\(d \in \Zp\) at time~\(t\), will be safe for the next \(d - 1\) time steps and then arrive at the edge's other endpoint at time \(t + d\).
Delays present a trade-off for the pursuer: they can be used to stretch the evader's possible appearances out over time, potentially reducing the vertices that can be occupied at individual time steps, but larger delays may also prolong a successful search.
We next state this setting more formally.

Given a graph \(G = (V, E)\), a \emph{strategy} for the pursuer comprises an assignment \(\tau \colon E \to \Zp\) of integer delays to each edge and a partial \emph{query function} \(q \colon \Zp \rightharpoonup V\).
If \(q(t)\)~is defined, it denotes the vertex inspected by the pursuer at time~\(t\); otherwise the pursuer passes at that time.
A strategy \((\tau, q)\) is \emph{winning} if either \(V = \emptyset\) or if it hits every walk in~\(G\) whose total delay is at least the strategy's \emph{duration} \(\max\supp(q)\): every such walk must admit a prefix with end vertex~\(v\) and total delay~\(t\) such that \(q(t) = v\).
If \(G\)~admits a winning strategy with duration~\(k\) for the pursuer, then we say that \(G\)~is \emph{solvable in time~\(k\)}.
Note that, by convention, no query at time~\(0\) is allowed, so that the pursuer cannot remove a starting vertex from the evader.
This limitation has the benefit that any winning strategy can be generalized to the setting where the pursuer cannot make any queries before a fixed time \(t_{\mathrm{p}} \in \Zp\): simply scale all delays and all query times by~\(t_{\mathrm{p}}\).

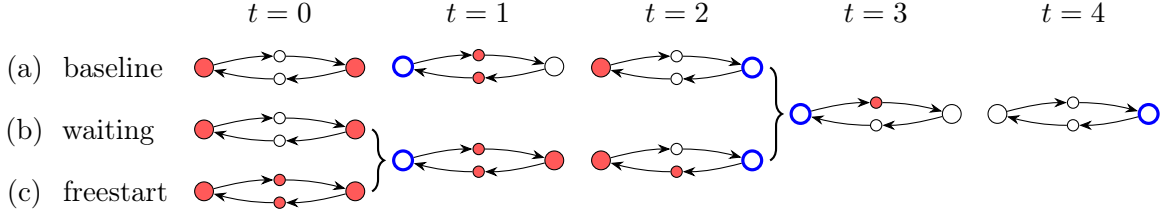
\begin{figure}
    \centering
    \newcommand{\ktwo}[4]{%
    \begin{tikzpicture}[baseline=(v_0.south)]
        \node[_vertex,         #1] (v_0) at (-1.00,  0.00) {};
        \node[_vertex, _small, #2] (e_0) at ( 0.00, +0.15) {};
        \node[_vertex,         #3] (v_1) at (+1.00,  0.00) {};
        \node[_vertex, _small, #4] (e_1) at ( 0.00, -0.15) {};

        \draw[_arc] (v_0) to[out=15,  in=180, looseness=0.9] (e_0);
        \draw[_arc] (e_0) to[out=0,   in=165, looseness=0.9] (v_1);
        \draw[_arc] (v_1) to[out=195, in=0,   looseness=0.9] (e_1);
        \draw[_arc] (e_1) to[out=180, in=-15, looseness=0.9] (v_0);
    \end{tikzpicture}%
}
\newcommand{\xxxxxxxxxxxxxxxx}{\phantom{\ktwo{}{}{}{}}}

\begin{tikzpicture}[
    b/.style={_burning},
    h/.style={_hit},
    brace/.style={
        decorate,
        decoration={brace, amplitude=4.5pt, raise=-1.25pt},
        thick,
    },
]
    \matrix[
        matrix of nodes,
        row sep=1mm,
        column sep=1mm,
        column 1/.style={column sep=0mm},
        column 2/.style={nodes={text width=3.6em, align=left}},
    ] {
        && \(t = 0\) & \(t = 1\) & \(t = 2\) & \(t = 3\) & \(t = 4\) \\
        (a) & baseline                  &
        \node (0a) {\ktwo{b}{ }{b}{ }}; &
        \node (1a) {\ktwo{h}{b}{ }{b}}; &
        \node (2a) {\ktwo{b}{ }{h}{ }}; &
        \node (3a) {\xxxxxxxxxxxxxxxx}; &
        \node (4a) {\xxxxxxxxxxxxxxxx}; \\
        (b) & waiting                   &
        \node (0b) {\ktwo{b}{ }{b}{ }}; &
        \node (1b) {\xxxxxxxxxxxxxxxx}; &
        \node (2b) {\xxxxxxxxxxxxxxxx}; &
        \node (3b) {\xxxxxxxxxxxxxxxx}; &
        \node (4b) {\xxxxxxxxxxxxxxxx}; \\
        (c) & freestart                 &
        \node (0c) {\ktwo{b}{b}{b}{b}}; &
        \node (1c) {\xxxxxxxxxxxxxxxx}; &
        \node (2c) {\xxxxxxxxxxxxxxxx}; &
        \node (3c) {\xxxxxxxxxxxxxxxx}; &
        \node (4c) {\xxxxxxxxxxxxxxxx}; \\
    };

    \node (1bc) at ($(1b.center)!0.5!(1c.center)$)
        {\ktwo{h}{b}{b}{b}};
    \node (2bc) at ($(2b.center)!0.5!(2c.center)$)
        {\ktwo{b}{ }{h}{b}};
    \node (3abc) at ($(3a.center)!0.5!($(3b.center)!0.5!(3c.center)$)$)
        {\ktwo{h}{b}{ }{ }};
    \node (4abc) at ($(4a.center)!0.5!($(4b.center)!0.5!(4c.center)$)$)
        {\ktwo{ }{ }{h}{ }};

    \draw[brace] (0b.east) -- (0c.east);
    \draw[brace] (2a.east) -- (2bc.east);
\end{tikzpicture}
    \caption{%
        A winning strategy for the pursuer on a single edge~\(e\) that works in all model variants studied.
        It assigns a travel time of \(\tau(e) = 2\) to the edge, which is represented by two directed paths of length~\(\tau(e)\) whose internal vertices (small) cannot be queried.
        Thick blue vertices mark queries by the pursuer while filled red vertices indicate possible evader locations.
        If the evader stays at or moves to a vertex before it is queried, they are caught.
        In setting~(b) the evader may wait at a vertex, otherwise they must move along the edge between time steps.
        In setting~(c) the evader started at an unknown time \(t_{\mathrm{e}} \leq 0\), so they could be traversing the edge at time~\(0\).
        For the combined setting, the outcome equals~(c) for this graph.
    }
    \label{fig:models}
\end{figure}

In \Cref{sec:waiting}, we give the evader the power to perform a \emph{lazy walk} that can wait at vertices, which can be modeled by a normal walk in an extension of~\(G\) where a loop with delay one is added to each vertex.
If the pursuer still has a winning strategy on a so-augmented graph, then we call~\(G\) \emph{waiting-solvable}.
The hunter and rabbit game with loops has a natural interpretation as a firefighting game, in which a burning graph needs to be extinguished vertex-by-vertex, and where vertices left burning reignite their neighbors.
We discuss this model further in \Cref{sec:related}.

Another power the evader may find useful is the ability to start their (lazy) walk at a time \(t_{\mathrm{e}} \leq 0\) unknown to the pursuer.
The evader then wins against a strategy \((\tau, q)\) if their walk has a total delay of at least \(\max\supp(q) - t_{\mathrm{e}}\) and is again not located at~\(q(t)\) at any time \(t \in \supp(q)\).
This differs from the setting in which the pursuer has to delay their first query until a time \(t_{\mathrm{p}} > 0\), as the pursuer is now unaware of the scaling factor \(1 - t_{\mathrm{e}}\) needed to adapt their strategy.
An early start at an unknown time can also be interpreted as the evader being allowed to start their walk inside an edge, though only at discrete positions and with a predetermined travel direction.
If the pursuer maintains a winning strategy on a graph~\(G\) in this setting, then we call~\(G\) \emph{freestart-solvable}.
In \Cref{sec:waiting-freestart}, we consider the pursuer's most difficult endeavor, where the evader may both wait at vertices and start ahead of time.
\Cref{fig:models} shows a strategy for a graph containing only one edge that is winning for all combinations of evader capabilities.

\subsection{Warm-up: directed graphs}

As a warm-up, we briefly consider the setting of directed graphs.
Here, the problem becomes much easier for the pursuer, who may now specify distinct delays for antiparallel arcs.
This yields a very quick and simple winning strategy, which the reader is invited to pause and ponder.

\begin{observation}
    Every directed graph is solvable in time~\(n\).
\end{observation}

\begin{proof}
    Let \(G = (V, A)\) be a directed graph on nodes \(V = [n]\).
    A winning strategy (\(\tau, q\)) for the pursuer assigns a delay of~\(\smash{\tau\bigl((i, j)\bigr)} \defas j\) to each arc \((i, j) \in A\) and queries \(q(j) \defas j\) at time \(j \in V\).
    Note that this strategy has a duration of \(\max\supp(q) = n\).
    Consider an evader whose first move is to traverse some arc \((i, j) \in A\).
    Then, they arrive at vertex~\(j\) at time \(\smash{\tau\bigl((i, j)\bigr)} = j = q(j)\).
\end{proof}

If the evader is allowed to wait at vertices, then the pursuer's strategy can be modified to catch them in quadratic time.

\begin{observation}
    Every directed graph is waiting-solvable in time \(n^2 + n\).
\end{observation}

\begin{proof}
    Let again \(G = (V, A)\) with \(V = [n]\) be a directed graph.
    To each arc \((i, j) \in A\), assign a delay of \(\smash{\tau\bigl((i, j)\bigr)} \defas j n + 1\).
    First, query \(q(t) \defas t\) for each \(t \in [n]\).
    This forces the evader to leave their starting vertex at some time \(t \in \set{0, \ldots, n - 1}\).
    Then, query \(q(j n + k) \defas j\) for all \(j, k \in [n]\).
    As \(j, k \geq 1\) and \(k \leq n\), all queries happen at distinct times, and so \(q \colon [n^2 + n] \to V\) is a well-defined query function with \(\max\supp(q) = n^2 + n\).
    Consider an evader whose first traversed arc is \((i, j) \in A\).
    As they entered this arc at some time \(t \in \set{0, \ldots, n - 1}\), they arrive at vertex~\(j\) at time \(t + \smash{\tau\bigl((i, j)\bigr)} = j n + t + 1\).
    As \(k \defas t + 1 \in [n]\), the evader is caught by the query \(q(j n + k) = j\) on arrival, and so \((\tau, q)\) is a winning strategy for the pursuer.
\end{proof}

These strategies are of no use in the undirected setting, as the pursuer is forced to assign the same delay to both directions of travel.
Thus, even if all edges have distinct delays and the evader cannot wait, they may arrive at either end of their first-traversed edge at the same time.

\subsection{Our results}%
\label{sec:results}

\begin{table}
    \centering
    \begin{tabular}{l@{\hspace{2em}}c@{\hspace{2em}}c}
        \toprule
            The evader \ldots
            & cannot stop
            & can wait at vertices \\
        \midrule
        must start at \(t = 0\)
            & \(\calO(m^3)\) [Thm.\,\ref{thm:baseline-general}]
            & \(\calO(n m^3)\) [Thm.\,\ref{thm:waiting-general}] \\[1ex]
        can start early
            & \multicolumn{2}{c}{\(\calO(8^n)\) [Thm.\,\ref{thm:waiting-freestart-general}]} \\
        \qquad on trees
            & \multicolumn{2}{c}{\(\calO(4^n)\) [Thm.\,\ref{thm:waiting-freestart-trees}]} \\
        \bottomrule
    \end{tabular}
    \caption{Time required to catch an evader on an undirected graph with \(n\)~vertices and \(m\)~edges.}
    \label{tab:results}
\end{table}

We show that the pursuer has a winning strategy on every graph, even in the most difficult setting where the evader may wait at vertices and start at an unknown time in the past.
More precisely, we show in \Cref{sec:baseline} that all graphs of order~\(n\) and size~\(m\) are solvable in time~\(\calO(m^3)\), in \Cref{sec:waiting} that they are waiting-solvable in time~\(\calO(nm^3)\), and in \Cref{sec:waiting-freestart} that they are waiting-freestart-solvable in time~\(\calO(8^n)\).
\Cref{sec:waiting-freestart-trees} describes a refined winning strategy for trees with a duration in~\(\calO(4^n)\).
These results are summarized in \Cref{tab:results}.

Both polynomial-duration winning strategies are based on pruning short walks.
In the baseline setting, half of all walks of length three are hit after traversing their first edge, part of the remaining walks after using their second edge, and the rest after passing through the third edge.
The queries are organized according to a lexicographical ordering of the edges, and they occur at distinct times as we assign edge lengths from a generalized Sidon set (see \Cref{sec:prelims}).
If waiting at vertices is allowed, we extend this strategy by scaling all edge delays by \(n + 1\) and interleaving an \(n\)-fold repetition of the original queries with new ones that compel the evader to move, preventing them from accumulating too much waiting time.

The exponential-duration strategies in \Cref{sec:waiting-freestart} both follow a recursive construction: given a graph for which a winning strategy exists, we introduce new vertices and edges with delays larger than that strategy's duration.
Each new edge then serves as a fuse, allowing us to execute the winning sub-strategy before an evader could have traversed it completely.
The new winning strategy then alternates batches of queries between the new and old vertices.

\subsection{Related work}%
\label{sec:related}

Pursuit--evasion games on graphs have been the subject of sustained investigation since the late 1970s.
Variants are broadly distinguished by the capabilities of the pursuers and the evader: slow or fast movement along edges versus teleportation, known locations versus invisibility, forced movement versus the ability to pass, randomized versus deterministic strategies.
They are routinely studied on both simple and non-simple, directed and undirected, and finite and infinite graphs.
For each setting, natural questions involve pinpointing the number of pursuers needed to win on all graphs of a given class and, conversely, identifying the graphs that are winnable by a given number of pursuers.
In the following, we discuss the games most closely related to ours: each pursuer can query an arbitrary vertex of a finite and undirected graph per time step while the evader can only move along the edges of the graph but is invisible to the pursuer.
For a more comprehensive treatment, we refer to the recent book of \citelong{Bonato2022} and the bibliography of \citelong{Fomin2008}.

\paragraph{Hunter and rabbit.}

In academic publications, the hunter and rabbit game outlined in \Cref{sec:intro} was independently introduced by \citelong{Britnell2013} and \citelong{Haslegrave2014}.
It made an earlier appearance, though, as Problem~6 of the Grade~9 division of the 2000 Moscow Mathematical Olympiad~\cite{Fedorov2011}.
The task was to first devise a winning strategy for the evader on a \mbox{\(2\)-subdivided} claw graph, and then to identify the graphs on which the evader wins unless any edge is removed.
None of these sources phrase the problem in terms of hunting a rabbit: \citeauthor{Britnell2013} help a prince find a challenge-posing princess in a palace, \citeauthor{Haslegrave2014} considers a cat chasing a mouse, and the Olympiad puzzle imagines a video game in which the player shoots at shelters connected by tunnels.
The first appearance of a rabbit in the pursuit-evasion literature appears to be in the work of \citelong{Adler2003}.
The authors consider a different model, though, in which the hunter moves along the edges of a graph (or stays put) according to a randomized strategy, while the rabbit either obeys the same rules or may jump to any vertex.

Both \citelong{Britnell2013} and \citelong{Haslegrave2014} characterize the graphs on which the hunter has a winning strategy, which are precisely the trees that exclude the Olympiad's \mbox{\(2\)-subdivided} claw as a subgraph.
Both works further describe an optimal strategy for such graphs, whose duration is linear in the size of a subtree that excludes certain leaves.
In response to the characterization, \citelong{Abramovskaya2016} ask how many hunters, each probing one vertex per round, are needed to catch the rabbit on an arbitrary graph.
The authors show that their \emph{hunting number} is at most one more than the graph's pathwidth, and that this bound is tight for every pathwidth above one.
They further study grid graphs, for which the hunting number is about half the size of the smaller side, and trees, where it is in \(\Omega(\log n / \log\log n) \cap \calO(\log n)\).
\citelong{Gruslys2015} tighten the latter result to~\(\Theta(\log n)\) with a precise upper bound of \(\lceil (1/2) \log_2 n \rceil\), and \citelong{Bolkema2019} give the exact hunting number of the \(d\)-dimensional hypercube, which is in \(\Theta(2^d / \sqrt{d})\).
Recently, \citelong{BenAmeur2026} show that the problem of deciding whether the hunting number of a given graph is below some threshold is \NP-hard, which partially answers a longstanding question of \citeauthor{Abramovskaya2016} about this problem's complexity.
This already holds if the graph is bipartite and, separately, even if the search is limited to just two time steps.
It remains open whether the general problem is contained in~\NP{}; it may well be \PSPACE-complete.

\paragraph{Firefighting.}

In the \emph{firefighting game}, all vertices of a graph are initially on fire, and in every round a number of firefighters can extinguish one vertex each.
Afterwards, the fire spreads from burning vertices to adjacent ones, including to vertices that have been extinguished before.
The central question is whether the firefighters can completely extinguish the graph.
Interpreting the burning vertices as possible locations of an invisible evader yields the hunter and rabbit game with loops attached to every vertex.
This seemingly minor modification changes the problem substantially: already a single edge cannot be cleared by just one firefighter.
\citelong{Bolkema2019} carry over their result about the hunting number of hypercubes to the firefighting setting, which they phrase in terms of a deaf rabbit that is not scared by the shots.
\citelong{Bernshteyn2022} put the game into focus and identify the ``inspection number'' of firefighters needed to be one more than the smaller side length of a grid, unbounded for subcubic trees, and at most one more than a graph's pathwidth.
They also study the same quantity when the pursuer is restricted to monotone strategies that prevent any reignition, as well as in a topological setting where the firefighters must be able to extinguish some further subdivision of any given subdivision of a graph.
\citelong{Althoetmar2025} give the game the name adopted here and show that two firefighters can extinguish precisely the caterpillar graphs.
They further prove \NP-hardness of deciding whether a graph is winnable by at most~\(k\) firefighters, and that shortest winning strategies may require superpolynomial time in the size of the graph.

The firefighting game differs from the \emph{firefighter problem} attributed to \citelong{Hartnell1995}, in which the fire initially occupies only a subset of the vertices, non-burning vertices can be permanently protected, and the goal is to minimize the number of vertices that catch fire~\cite{Finbow2009,Wagner2021}.

\paragraph{Intermediate models.}

Letting the evader perform a walk on a non-simple graph interpolates between the hunter and rabbit game (no loops) and the firefighting game (loops on every vertex).
\citelong{BenAmeur2024} do so in the context of directed graphs and phrase the result as a variant of \emph{cops and robbers}.
Here, the ``cop number'' of pursuers needed is at most one more than a graph's pathwidth, is no larger than the size of a feedback vertex set, and is \NP-hard to compute even in the loopless setting.
\citelong{BenAmeur2026} give a forbidden-subgraph characterization of non-simple undirected graphs that are winnable by a single hunter.
It excludes four subgraphs with loops in addition to cycles and the \mbox{\(2\)-subdivided} claw already excluded for simple graphs.

\section{Preliminaries}%
\label{sec:prelims}

We write \(\Zn\) for the non-negative and~\(\Zp\) for the positive integers.
For \(n \in \Zp\), we define the range \([n] \defas \set{k \in \Zp \mid k \leq n}\) and the extended range \(\langle{n}\rangle \defas \set{0} \cup [n]\).
For a singleton set \(X = \set{x}\), we write \(\ast X \defas x\) for its element.
All graphs in the following are finite, undirected, and simple.
When a graph \(G = (V, E)\) is clear from the context, then we write \(n \defas |V|\) for its order and \(m \defas |E|\) for its size.
We further assume that vertices are labeled as \(V = [n]\).

The following notion is used to track evaders that can wait at vertices.

\begin{definition}[Lazy walk]
    Let \(G = (V, E)\) be a graph with edge delays \(d \colon E \to \Zp\).
    A \emph{lazy walk} \(W = (v_1, \ldots, v_k)\) \emph{with waiting times} \((\alpha_1, \ldots, \alpha_k)\) and \emph{start time}~\(t_0 \in \Z\) is a walk~\(W\) in~\(G\) that is located at vertex~\(v_i\) at all times~\(t \in \Z\) such that
    \[
        0 \leq t - t_0 - \sum_{j=1}^{i-1} \bigl( \alpha_j + d(\set{v_j, v_{j+1}}) \bigr) \leq \alpha_i.
    \]
\end{definition}

By the \emph{length} of a lazy walk~\(W\), we mean its topological length \(k - 1\).

For distinguishing travel times along distinct sets of edges, we will assign delays from sets with the following property.

\begin{definition}[Generalized \(B_h\)-sets]\label{def:B_h}
    A set of \(m\)~integers \(A = \set{a_1, \ldots, a_m} \subset \Zn\) is a \emph{(finite) \(B_{g,h}^r\)-set} if for every \(c, d \in \Zn^m\) with
    \[
        g \leq \lVert c \rVert_1, \lVert d \rVert_1 \leq h
        \qquad\text{and}\qquad
        \lVert c \rVert_\infty, \lVert d \rVert_\infty \leq r,
    \]
    we have that
    \[
        \sum_{i=1}^m c_i a_i = \sum_{i=1}^m d_i a_i \implies c = d.
    \]
\end{definition}

In other words, a \(B_{g,h}^r\)-set is a set for which any multi-subset with cardinality between \(g\) and~\(h\) and with multiplicity at most~\(r\) has a unique sum.

\begin{example}\label{exm:powers}
    For \(b \in \Z_{\geq 2}\), the powers \(P_b \defas \set{b^k \mid k \in \Zn}\) form a \(B_{0,b}^{b-1}\)-set but not a \(B_{0,b}^{b}\)-set.
    The former can be seen from the \(b\)-ary representation of sums involving fewer than~\(b\) powers of~\(b\), the latter follows from \(b \cdot b^k = 1 \cdot b^{k+1}\).
\end{example}

Notice that the \(B_{h,h}^h\)-sets are exactly the \emph{\(B_h\)-sets} (\(B_2\)-sets are also known as \emph{Sidon sets}), while the \(B_{0,m}^1\)-sets are precisely the \emph{subset-sum-distinct} (SSD) sets.
We only make use of \(B_{0,3}^2\)-sets in this work.
These are more restricted than \(B_3\)~sets in the sense that sums of less than three (possibly repeated) elements are considered for the uniqueness condition, but less restricted insofar as sums of the form~\(3 a_i\) are not considered.

The following is a straightforward way to obtain a \(B_{g,h}^r\)-set from a \(B_h\)-set for any \(g, r \leq h\).

\begin{restatable}{lemma}{bhshifting}\label{lem:B_h-shifting}
    Let \(A\)~be a \(B_h\)-set for some \(h \in \Zp\).
    Then, \(\set{a - \min A \mid a \in A} \setminus \set{0}\) is a \(B_{g,h}^r\)-set of size~\(|A| - 1\) for any \(g, r \leq h\).
\end{restatable}

\begin{proof}
    See \Cref{sec:deferred}.
\end{proof}

The extremal properties of \(B_h\)~sets, specifically the ratio between the cardinality and the diameter of a set, are well-studied in the literature.
We will make use of a classic result.

\begin{theorem}[\citelong{Bose1962}]
    For any prime power~\(m\), there exists a \(B_h\)-set~\(A\) with \(|A| = m\) and \(\max A < m^h\).
\end{theorem}

Together with \Cref{lem:B_h-shifting} and the Bertrand--Chebyshev theorem, which guarantees that prime powers are less than a factor of two apart, we obtain the following.

\begin{lemma}\label{lem:B_3}
    For any~\(m \in \Zp\), there exists a \(B_{0,3}^2\)-set~\(A\) with \(|A| = m\) and \(\max A \in \calO(m^3)\).
\end{lemma}

\section{A strategy for undirected graphs}%
\label{sec:baseline}

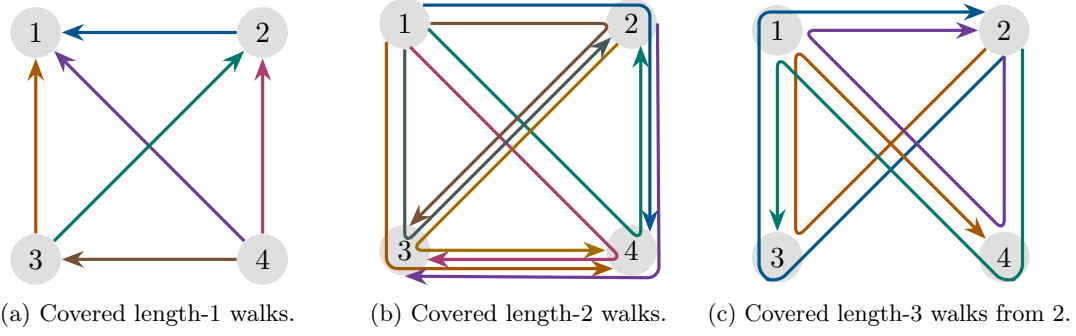
\begin{figure}
    \tikzset{node distance=3cm}
    \centering
    \begin{subfigure}{0.3\textwidth}
        \centering
        \begin{tikzpicture}
    \node[vertex] (1) {\(1\)};
    \node[vertex, right of=1] (2) {\(2\)};
    \node[vertex, below of=1] (3) {\(3\)};
    \node[vertex, right of=3] (4) {\(4\)};

    \draw[walk, walk1] (2) to (1);
    \draw[walk, walk2] (3) to (1);
    \draw[walk, walk3] (4) to (1);
    \draw[walk, walk4] (3) to (2);
    \draw[walk, walk5] (4) to (2);
    \draw[walk, walk6] (4) to (3);
\end{tikzpicture}%
        \subcaption{Covered length-\(1\) walks.}
    \end{subfigure}
    \begin{subfigure}{0.3\textwidth}
        \centering
        \begin{tikzpicture}
    \node[vertex] (1) {\(1\)};
    \node[vertex, right of=1] (2) {\(2\)};
    \node[vertex, below of=1] (3) {\(3\)};
    \node[vertex, right of=3] (4) {\(4\)};

    \draw[walk, walk1] 
        (1.north east) to
        (2.north east) to
        (4.north east);
    \draw[walk, walk2] 
        (1.south west) to
        (3.south west) to
        (4.south west);
    \draw[walk, walk3] 
        (2.east) to
        ($(4.south east)!-0.5!(4.center)$) to
        (3.south);

    \draw[walk, walk4] 
        ($(1.south east)!0.7!(1.east)$) to
        ($(4.north east)!0.5!(4.center)$) to
        ($(2.south)!0.5!(2.south east)$);
    \draw[walk, walk5] 
        ($(1.south east)!0.7!(1.south)$) to
        ($(4.south west)!0.5!(4.center)$) to
        ($(3.east)!0.5!(3.south east)$);

    \draw[walk, walk6] 
        ($(1.east)$) to
        ($(2.west)!0.25!(2.center)$) to
        ($(3.north east)!0.75!(3.north)$);
    \draw[walk, walk7] 
        ($(2.south west)!0.25!(2.south)$) to
        ($(3.east)!0.75!(3.center)$) to
        ($(4.west)$);
    \draw[walk, walk8] 
        (1) to
        ($(3.center)!0.25!(3.north)$) to
        ($(2.south west)!0.25!(2.west)$);
\end{tikzpicture}%
        \subcaption{Covered length-\(2\) walks.}
    \end{subfigure}
    \begin{subfigure}{0.3\textwidth}
        \centering
        \begin{tikzpicture}
    \node[vertex] (1) {\(1\)};
    \node[vertex, right of=1] (2) {\(2\)};
    \node[vertex, below of=1] (3) {\(3\)};
    \node[vertex, right of=3] (4) {\(4\)};

    \draw[walk, walk1] 
        (2.south) to
        (3.south) to
        (3.south west) to
        (1.north west) to
        (2.north west);
    \draw[walk, walk2] 
        (2.south west) to
        (3.north east) to
        (1.south east) to
        (4.north west);
    \draw[walk, walk3] 
        (2.south) to
        (4.north) to
        (1.east) to
        (2.west);
    \draw[walk, walk4] 
        (2.south east) to
        (4.south east) to
        (4.south) to
        (1.south) to
        (3.north);
\end{tikzpicture}%
        \subcaption{Covered length-\(3\) walks from~\(2\).}
    \end{subfigure}
    \caption{%
        The strategy of \Cref{thm:baseline-general} on the complete graph~\(K_4\).
        It covers all walks of length~\(1\) with decreasing vertices, all walks of length~\(2\) with non-decreasing edges according to a lexicographical order, and all remaining walks of length~\(3\).
        Covered walks that go back-and-forth along an edge or that start in vertex~\(3\) with length~\(3\) are not shown.
        Delays are chosen such that all these walks have distinct lengths.
    }
    \label{fig:baseline-general}
\end{figure}

Our first main result states that the pursuer has a winning strategy of polynomial duration on any undirected graph.
In the proof, we assume an arbitrary ordering of the vertices, and we derive a lexicographical ordering of the edges from it.
We then ensure via queries that in any evasive walk, the first two vertices must be visited in increasing order, and the first two edges must be visited in decreasing order.
This guarantees that any walk of length three that uses a given multiset of edges has a unique end vertex, which can be queried to end the walk.
See \Cref{fig:baseline-general} for an example.
To make all necessary queries occur at distinct times, delays are chosen from a \smash{\(B_{0,3}^2\)}-set, ensuring all relevant combinations have distinct sums.

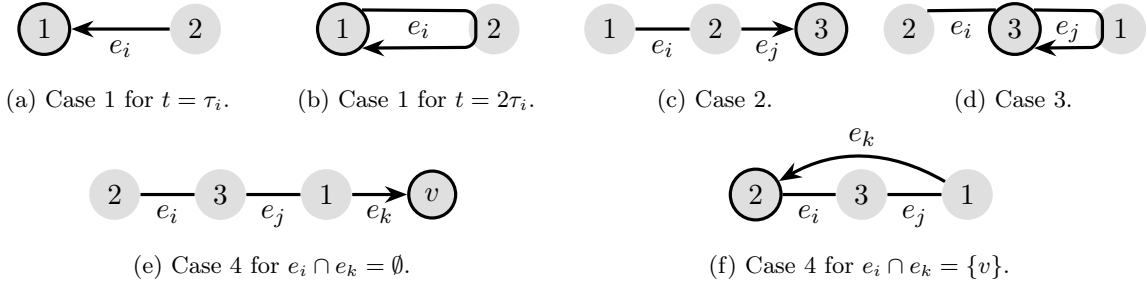
\begin{figure}
    \centering
    \begin{subfigure}{0.24\textwidth}
        \centering
        \begin{tikzpicture}[
            node distance=20mm,
        ]
            \node[vertex, hit] (1) {\(1\)};
            \node[vertex, right of=1] (2) {\(2\)};
            \draw[walk] (2) to node[midway, below] {\(e_i\)} (1);
            \path[walk] (1) to node[midway, below] {\phantom{\(e_j\)}} (2);
        \end{tikzpicture}
        \subcaption{Case~\(1\) for \(t = \tau_i\).}
    \end{subfigure}
    \begin{subfigure}{0.24\textwidth}
        \centering
        \begin{tikzpicture}[
            node distance=20mm,
        ]
            \node[vertex, hit] (1) {\(1\)};
            \node[vertex, right of=1] (2) {\(2\)};
            \draw[walk] (1.north east) to node[midway, below] {\(e_i\)} (2.north west) to (2.south west) to (1.south east);
            \path[walk] (1) to node[midway, below] {\phantom{\(e_j\)}} (2);
        \end{tikzpicture}
        \subcaption{Case~\(1\) for \(t = 2 \tau_i\).}
    \end{subfigure}
    \begin{subfigure}{0.24\textwidth}
        \centering
        \begin{tikzpicture}[
            node distance=14mm,
        ]
            \node[vertex] (1) {\(1\)};
            \node[vertex, right of=1] (2) {\(2\)};
            \node[vertex, hit, right of=2] (3) {\(3\)};
            \draw[walk] (1) to node[midway, below] {\(e_i\)} (2) to node[midway, below] {\(e_j\)} (3);
        \end{tikzpicture}
        \subcaption{Case~\(2\).}
    \end{subfigure}
    \begin{subfigure}{0.24\textwidth}
        \centering
        \begin{tikzpicture}[
            node distance=14mm,
        ]
            \node[vertex] (2) {\(2\)};
            \node[vertex, hit, right of=2] (3) {\(3\)};
            \node[vertex, right of=3] (1) {\(1\)};
            \draw[partialwalk] (2.north east) to node[midway, below] {\(e_i\)} (3.north west);
            \draw[walk] (3.north east) to node[midway, below] {\(e_j\)} (1.north west) to (1.south west) to (3.south east);
            \path[walk] (2) to node[midway, below] {\phantom{\(e_j\)}} (3);
        \end{tikzpicture}
        \subcaption{Case~\(3\).}
    \end{subfigure}
    \begin{subfigure}{0.48\textwidth}
        \centering
        \begin{tikzpicture}[
                node distance=14mm,
            ]
            \node[vertex] (2) {\(2\)};
            \node[vertex, right of=2] (3) {\(3\)};
            \node[vertex, right of=3] (1) {\(1\)};
            \node[vertex, hit, right of=1] (v) {\(v\)};
            \draw[walk] (2) to node[midway, below] {\(e_i\)} (3) to node[midway, below] {\(e_j\)} (1) to node[midway, below] {\(e_k\)} (v);
        \end{tikzpicture}
        \subcaption{Case~\(4\) for \(e_i \cap e_k = \emptyset\).}
    \end{subfigure}
    \begin{subfigure}{0.48\textwidth}
        \centering
        \begin{tikzpicture}[
            node distance=14mm,
        ]
            \node[vertex, hit] (2) {\(2\)};
            \node[vertex, right of=2] (3) {\(3\)};
            \node[vertex, right of=3] (1) {\(1\)};
            \draw[partialwalk] (2) to node[midway, below] {\(e_i\)} (3);
            \draw[walk] (3) to node[midway, below] {\(e_j\)} (1) to[bend right] node[midway, above] {\(e_k\)} (2);
        \end{tikzpicture}
        \subcaption{Case~\(4\) for \(e_i \cap e_k = \set{v}\).}
    \end{subfigure}
    \caption{%
        Examples of walks covered by the four cases of the query function~\(q\) in the proof of \Cref{thm:baseline-general}.
        The outlined vertex is the one queried.
        Notice that the edge order is \(\set{2, 1} \prec \set{3, 1} \prec \set{3, 2}\).
    }
    \label{fig:baseline-general-q}
\end{figure}

\begin{theorem}\label{thm:baseline-general}
    All graphs are solvable in time~\smash{\(\calO\bigl(m^3\bigr)\)}.
\end{theorem}

\begin{proof}
    Let \(G = (V, E)\) be a graph.
    We may assume that \(E \neq \emptyset\); otherwise, the pursuer wins by definition if \(V = \emptyset\), or else by querying an arbitrary vertex at time~\(1\), since every walk in~\(G\) has total delay~\(0\).
    Order the edges lexicographically with the greater vertex written first:
    for \(\{u^+, u^-\}, \{v^+, v^-\} \in E\) with \(u^+ > u^-\) and \(v^+ > v^-\), set
    \[
        \{u^+, u^-\} \prec \{v^+, v^-\}
        \quad\Longleftrightarrow\quad
        u^+ < v^+ \lor (u^+ = v^+ \land u^- < v^-).
    \]
    Let further \(\pi \colon [m] \to E\) be the unique permutation such that \(\pi(i) \prec \pi(j)\) if and only if \(i < j\), and denote the \(i\)-th edge by \(e_i \defas \pi(i)\) for all \(i \in [m]\).
    To obtain the edge delays, let \(T = \set{\tau_1, \ldots, \tau_m}\) be a \(B_{0,3}^2\)-set with \(\max T \in \calO(m^3)\), whose existence is guaranteed by~\Cref{lem:B_3}, and assign \(\tau(e_i) \defas \tau_i\) for all \(i \in [m]\).
    Consider finally the query function \(q \colon \Zp \rightharpoonup V\) given by
    \[
        q(t) \defas \begin{cases}
            \min e_i, & \exists i \in [m] \colon t \in \set{\tau_i, 2\tau_i}, \\
            v, & \exists i < j \in [m], v \in V \colon t = \tau_i + \tau_j \land e_j \setminus e_i = \set{v}, \\
            v, & \exists i > j \in [m], v \in V \colon t = \tau_i + 2 \tau_j \land e_i \cap e_j = \set{v}, \\
            v, & \begin{aligned}[t]
                &\exists k \neq i > j \in [m], v \in V \colon \\
                &\quad
                t = \tau_i + \tau_j + \tau_k
                \land e_i \cap e_j = \set{\max e_i} \not\subseteq e_k
                \land e_k \setminus e_j = \set{v}.
            \end{aligned}
        \end{cases}
    \]
    The cases are illustrated in \Cref{fig:baseline-general-q}.

    To see that \(q\)~is well-defined, first observe that the constraints on the time~\(t\) and the fact that \(T\)~is a \(B_{0,3}^2\)-set (\Cref{def:B_h}) render all four cases disjoint.
    Specifically, the fourth case is disjoint from the third case as the condition of \(e_k \setminus e_j = \set{v}\) implies \(k \neq j\), so that \(\tau_i + \tau_j + \tau_k \neq \tau_i + 2 \tau_j\).
    It remains to show that the choice of~\(i\) in the first case and the choices of~\(v\) in the last three cases are unique.
    For the first case this follows again from \Cref{def:B_h}, as \(\tau_i\), \(\tau_j\), \(2 \tau_i\), and \(2 \tau_j\) are all distinct for any \(i \neq j \in [m]\).
    To see uniqueness in the second case, let \(i < j \in [m]\) and \(v \in V\) with \(t = \tau_i + \tau_j\) and \(e_j \setminus e_i = \set{v}\).
    Suppose that there are \(i' < j' \in [m]\) and \(v' \in V\) with \(t = \tau_{i'} + \tau_{j'}\) and \(e_{j'} \setminus e_{i'} = \set{v'}\) but \(v' \neq v\).
    From \Cref{def:B_h} we obtain \(\set{i, j} = \set{i', j'}\), and from \(i < j\) and \(i' < j'\) it follows that \(i = i'\) and \(j = j'\).
    But then, \(\set{v'} = e_{j'} \setminus e_{i'} = e_j \setminus e_i = \set{v}\) contradicts \(v' \neq v\).
    For the third case, let \(i > j \in [m]\) and \(v \in V\) with \(t = \tau_i + 2 \tau_j\) and \(e_i \cap e_j = \set{v}\).
    Suppose that there are \(i' > j' \in [m]\) and \(v' \in V\) with \(t = \tau_{i'} + 2 \tau_{j'}\) and \(e_{i'} \cap e_{j'} = \set{v'}\) but \(v' \neq v\).
    Here, \Cref{def:B_h} directly gives us \(i' = i\) and \(j' = j\), so that \(\set{v'} = e_{i'} \cap e_{j'} = e_i \cap e_j = \set{v}\) again contradicts \(v' \neq v\).

    Uniqueness is least obvious in the fourth case.
    Let \(i > j, k \in [m]\) and \(v \in V\) such that \(t = \tau_i + \tau_j + \tau_k\), \(e_i \cap e_j = \set{\max e_i} \not\subseteq e_k\), and \(e_k \setminus e_j = \set{v}\).
    Suppose towards a contradiction that there are \(i' > j', k' \in [m]\) and \(v' \in V\) such that \(t = \tau_{i'} + \tau_{j'} + \tau_{k'}\), \(e_{i'} \cap e_{j'} = \set{\max e_{i'}}\), and \(e_{k'} \setminus e_{j'} = \set{v'}\), but \(v' \neq v\).
    Note that due to \(k \neq i > j\) and \(e_k \setminus e_j = \set{v}\), we have that \(i\), \(j\), and \(k\) are pairwise distinct.
    By analogy, so are \(i'\), \(j'\), and \(k'\).
    From this and \(\tau_i + \tau_j + \tau_k = t = \tau_{i'} + \tau_{j'} + \tau_{k'}\), \Cref{def:B_h} gives us \(\set{i', j', k'} = \set{i, j, k}\).
    In the following we show that further \((i', j', k') = (i, j, k)\).
    This then implies \(\set{v'} = e_{k'} \setminus e_{j'} = e_k \setminus e_j = \set{v}\), contradicting \(v' \neq v\).

    First suppose towards a contradiction that \(i' = j\).
    As \(j' = i\) would imply \(i' = j < i = j'\), contradicting \(i' > j'\), it must be that \((i', j', k') = (j, k, i)\).
    From
    \[
        e_i \cap e_j = \set{\max e_i},
        \quad
        e_j \cap e_k = e_{i'} \cap e_{j'} = \set{\max e_{i'}} = \set{\max e_j},
        \quad\text{and}\quad
        e_i \setminus e_k = e_{k'} \setminus e_{j'} = \set{v'},
    \]
    it follows that \(\set{e_i, e_j, e_k}\) induces a triangle with vertices \(\set{\max e_i, \max e_j, v}\) in~\(G\): the third vertex can be derived from \(e_k \setminus e_j = \set{v}\).
    In particular, this implies \(\max e_i \neq \max e_j\).
    If \(\max e_i < \max e_j\), then \(e_i \prec e_j\) would contradict \(i > j\).
    Thus, it must be that \(\max e_i > \max e_j\).
    But from \(\max e_j \neq \max e_i \in e_i \cap e_j\) we further have \(e_j = \set{\max e_i, \max e_j}\), and by extension \(\max e_i = \min e_j < \max e_j\), a contradiction.

    Next suppose that \(i' = k\).
    It could be that \((i', j', k') = (k, i, j)\) or that \((i', j', k') = (k, j, i)\).
    Suppose the former, then we have that
    \[
        e_i \cap e_j = \set{\max e_i},
        \quad
        e_k \cap e_i = e_{i'} \cap e_{j'} = \set{\max e_{i'}} = \set{\max e_k},
        \quad\text{and}\quad
        e_k \setminus e_j = \set{v}.
    \]
    Thus, \(\set{e_i, e_j, e_k}\) again induces a triangle in~\(G\), this time with vertices \(\set{\max e_i, \max e_k, u}\) where \(\set{u} = e_j \cap e_k\).
    Since \(\max e_i \in e_i \cap e_j\), \(u \in e_j \cap e_k\), and \(\max e_k \in e_k\), we have that \(e_j = \set{\max e_i, u}\) and \(e_k = \set{u, \max e_k}\), and thus \(e_i = \set{\max e_k, \max e_i}\).
    It follows that \(\max e_k = \min e_i < \max e_i\).
    But this implies \(e_k \prec e_i\), contradicting \(k = i' > j' = i\).

    Consider now the case of \((i', j', k') = (k, j, i)\), which is the remaining case assuming \(i' = k\).
    We find that
    \[
        e_i \cap e_j = \set{\max e_i}
        \quad\text{and}\quad
        e_k \cap e_j = e_{i'} \cap e_{j'} = \set{\max e_{i'}} = \set{\max e_k},
    \]
    as well as
    \[
        e_k \setminus e_j = \set{v}
        \neq
        \set{v'} = e_{k'} \setminus e_{j'} = e_i \setminus e_j,
    \]
    so that \(\set{e_i, e_j, e_k}\) induces the path \((v', \max e_i, \max e_k, v)\) in~\(G\).
    In particular, we have that \(e_j = \set{\max e_i, \max e_k}\).
    If \(\max e_k > \max e_i\), then \(\max e_j = \max e_k > \max e_i\) and thus \(e_i \prec e_j\), contradicting \(i > j\).
    If \(\max e_i > \max e_k\), then \(\max e_j = \max e_i > \max e_k\) and thus \(e_k \prec e_j\), contradicting \(k = i' > j' = j\).

    We now have that \(i' = i\), so suppose that \(j' = k\), implying \((i', j', k') = (i, k, j)\).
    In this case, we have
    \[
        e_i \cap e_k = e_{i'} \cap e_{j'} = \set{\max e_{i'}} = \set{\max e_i},
    \]
    which contradicts \(\set{\max e_i} \not\subseteq e_k\).
    This concludes the argument that \((i', j', k') = (i, j, k)\) and thus \(v' = v\).
    We thus established that \(q\)~is a well-defined partial function.

    We next show that any walk of length at least three is hit by a query.
    Let \(P = (a, b, c, d)\) be a walk in~\(G\) along the edges
    \[
        e_i = \set{a, b},
        \quad
        e_j = \set{b, c},
        \quad\text{and}\quad
        e_k = \set{c, d},
    \]
    for some \(i, j, k \in [m]\).
    Suppose towards a contradiction that neither
    \[
        q(\tau_i) = b,
        \quad\text{nor}\quad
        q(\tau_i + \tau_j) = c,
        \quad\text{nor}\quad
        q( \tau_i + \tau_j + \tau_k) = d.
    \]
    Then, we have that \(a < b\), as otherwise \(q(\tau_i) = \min e_i = b\).
    We further have that \(i \neq j\), as otherwise \(q(\tau_i + \tau_j) = q(2 \tau_i) = \min e_i = a = c\).
    Additionally, we know that \(i > j\).
    To see this, observe that \(|e_j \setminus e_i| = 0\) would imply \(i = j\), while \(|e_j \setminus e_i| = 2\) contradicts \(b \in e_i \cap e_j\).
    We have thus \(|e_j \setminus e_i| = 1\) and specifically \(e_j \setminus e_i = \set{c}\).
    But then, \(i < j\) would imply \({q(\tau_i + \tau_j)} = c\) by the second case of~\(q\).
    We thus established that \(a < b\) and \(i > j\), which together imply that \(e_i \cap e_j = \set{b} = \set{\max e_i}\).
    Consider now the last vertex~\(d\).
    As \(e_k = \set{c, d}\), we know that \(d \neq c\).
    If \(d = b\), then \(e_j = e_k\) and thus \(q(\tau_i + \tau_j + \tau_k) = q(\tau_i + 2 \tau_j) = b = d\) by the third case of~\(q\).
    Otherwise, \(\set{e_i, e_j, e_k}\) induces either a triangle (if \(d = a\)) or a path in~\(G\).
    In both cases, it follows that \(e_i \cap e_j \cap e_k = \emptyset\), implying \(e_i \cap e_j \not\subseteq e_k\), and that \(e_k \setminus e_j = \set{d}\).
    But then, we have that \(q(\tau_i + \tau_j + \tau_k) = d\) by the fourth case of~\(q\).

    Let now \(W\)~be an unhit walk with total delay~\(t_W\), and let \(\ell \defas \argmax_{i \in [m]} \tau_i\) be the index of the unique edge with maximum delay.
    As \(W\)~has length at most two, we have that \(t_W \leq 2 \tau_\ell\).
    Equality implies that \(W\)~must traverse \(e_\ell\)~twice, but we already ruled this out above; so \(t_W < 2 \tau_\ell\).
    Since \(q(2 \tau_\ell) = \min e_\ell\), we have \(t_W < \max\supp(q)\), and so \((\tau, q)\) is indeed a winning strategy for the pursuer.
    It only remains to bound its duration from above by
    \[
        \max\supp(q)
        \leq \max_{i,j,k \in [m]} \tau_i + \tau_j + \tau_k
        = 3 \max T
        \in \calO\bigl(m^3\bigr).
        \qedhere
    \]
\end{proof}

\section{Catching a patient evader}%
\label{sec:waiting}

In \Cref{sec:baseline}, we showed how to catch an evader who leaves a vertex as soon as they arrive.
This is precisely the \emph{hunter and rabbit} setting with only one hunter and an invisible rabbit.
Are we still able to get a hold of a perfectly camouflaged rabbit if it manages to remain calm at a vertex, delaying its exit by any number of time steps?
Our second theorem states that this is indeed possible, and with only a linear overhead in the duration of the pursuer's strategy.

\begin{theorem}\label{thm:waiting-general}
    All graphs are waiting-solvable in time~\smash{\(\calO\bigl(nm^3\bigr)\)}.
\end{theorem}

The strategy extends that of \Cref{thm:baseline-general}.
To account for waiting times, we scale all delays by a factor of \(s \defas n + 1\).
The new query function~\(q'\) is then derived from~\(q\) as follows.
First, we query each vertex once, that is \(q'(i) \defas i\) for all \(i \in [n]\).
Then, for each query~\(q(t)\) at time~\(t\) in the original strategy, we make \(n\)~queries \(q'(st + i - 1) \defas q(t)\) for all \(i \in [n]\), followed by a single query, at time \(st + n\), on the opposite end of the walk that was covered by~\(q(t)\).
The initial and single queries ensure that the evader cannot accumulate a waiting time of more than \(n - 1\) without being caught.
But then, they arrive in time to be hit by the \(n\)-fold repetition of the original strategy~\(q\).

\begin{proof}[Proof of \Cref{thm:waiting-general}.]
    Label the vertices \(V = [n]\) arbitrarily and order the edges lexicographically as in the proof of \Cref{thm:baseline-general}.
    Let further \(s \defas n + 1\) be a scaling factor, assign a delay of \(\tau'(e_i) \defas s \tau_i\) for all \(i \in [m]\), and consider the query function \(q' \colon \Zp \rightharpoonup V\) given by
    \[
        q'(t) \defas \begin{cases}
            t, & t \in [n], \\
            q(\lfloor t / s \rfloor), & t > n \land t \bmod s \neq n, \\
            r(\lfloor t / s \rfloor), & t > n \land t \bmod s = n,
        \end{cases}
    \]
    where \(q\) and~\(\tau_i\) are defined as in the proof of \Cref{thm:baseline-general}, and where \(r\)~is given by
    \[
        r(t) \defas \begin{cases}
            \max e_i, & \exists i \in [m] \colon t = \tau_i, \\
            \ast(e_i \setminus e_j), & \exists i < j \in [m], v \in V \colon t = \tau_i + \tau_j \land e_j \setminus e_i = \set{v}.
        \end{cases}
    \]
    Intuitively, \(r\)~covers the opposite end of the walks of length one and two that are covered by~\(q\).
    Note that the definition of \(q'(t)\) may refer to \(q\) and~\(r\) at undefined points; with this we mean that \(q'\)~is itself undefined at~\(t\).
    As the cases in the definition of~\(r\) are a subset of the cases of~\(q\), it follows that also~\(r\), and by extension~\(q'\), are well-defined partial functions.

    We first show that any lazy walk of length at most two that waits for at least \(n\) time steps is hit by a query.
    If the walk has length zero, then the evader simply waits for \(n\) or more time steps at their chosen starting vertex~\(v \in V\), and is found at time~\(v \leq n\) by the first case in the definition of~\(q'\); formally \(q'(v) = v\).
    For a lazy walk of length one, consider \(W = (a, b)\) with waiting times \((\alpha, \beta)\) such that \(\alpha + \beta \geq n\), and let \(i \in [m]\) with \(e_i = \set{a, b}\).
    Suppose towards a contradiction that \(W\)~is not hit by~\(q'\), that is
    \[
        \forall w \in \langle\alpha\rangle \colon q'(w) \neq a
        \qquad\text{and}\qquad
        \forall w \in \langle\beta\rangle \colon q'(\alpha + s \tau_i + w) \neq b.
    \]
    Then, we have that \(\alpha < n < s\), as otherwise \(a \in \langle\alpha\rangle\) and \(q'(a) = a\), as for the lazy walk of length zero.
    We distinguish two cases based on the direction of the walk along~\(e_i\).
    If \(b = \min e_i\), then choosing \(w = 0 \in \langle\beta\rangle\) gives
    \[
        \alpha + s \tau_i + w \geq s > n
        \qquad\text{and}\qquad
        (\alpha + s \tau_i + w) \bmod s = \alpha < n,
    \]
    implying the contradiction that
    \[
        q'(\alpha + s \tau_i + w)
        = q\left(\left\lfloor \frac{\alpha + s \tau_i + w}{s} \right\rfloor\right)
        = q\left(\left\lfloor \frac{\alpha}{s} + \tau_i \right\rfloor\right)
        = q(\tau_i)
        = \min e_i
        = b.
    \]
    If \(b = \max e_i\), then picking \(w = n - \alpha \in \langle\beta\rangle\) yields
    \[
        \alpha + s \tau_i + w \geq s > n
        \qquad\text{and}\qquad
        (\alpha + s \tau_i + w) \bmod s = \alpha + w = n,
    \]
    and thus another contradiction is reached as
    \[
        q'(\alpha + s \tau_i + w)
        = r\left(\left\lfloor \frac{\alpha + s \tau_i + w}{s} \right\rfloor\right)
        = r\left(\left\lfloor \frac{n}{s} + \tau_i \right\rfloor\right)
        = r(\tau_i)
        = \max e_i
        = b.
    \]

    Let now \(W = (a, b, c)\) be a lazy walk of length two with waiting times \((\alpha, \beta, \gamma)\), where \(\alpha + \beta + \gamma \geq n\).
    Let further \(i, j \in [m]\) index the edges \(e_i = \set{a, b}\) and \(e_j = \set{b, c}\).
    Suppose towards a contradiction that \(W\)~is not hit by~\(q'\), meaning that
    \[
        \forall w \in \langle\alpha\rangle \colon q'(w) \neq a,
        \qquad
        \forall w \in \langle\beta\rangle \colon q'(\alpha + s \tau_i + w) \neq b,
        \qquad\text{and}\qquad
        \forall w \in \langle\gamma\rangle \colon q'(t_w) \neq c,
    \]
    where we abbreviate the last time as
    \[
        t_w \defas \alpha + \beta + w + s \tau_i + s \tau_j.
    \]
    We have again that \(\alpha < n < s\), as otherwise \(a \in \langle\alpha\rangle\) and \(q'(a) = a\).
    Similarly, we know that \(\alpha + \beta < n\), or else \(q'(\alpha + s \tau_i + w) = b\) for some \(w \in \langle\beta\rangle\) as in the previous case.
    We may also assume that \(a = \min e_i\), as otherwise \(\alpha + s \tau_i + w \geq s > n\) and \((\alpha + s \tau_i + w) \bmod s = \alpha < n\) for \(w = 0\) would imply \(q'(\alpha + s \tau_i + w) = \min e_i = b\).
    We next distinguish two cases based on the topology of the walk.
    If \(a = c\), meaning that \(W\)~goes back and forth along~\(e_i = e_j\), we have that \(W\)~is covered by~\(q'\) for \(w = 0 \in \langle\gamma\rangle\), as
    \begin{equation}\label{eq:wg-1}
        t_w \geq s > n
        \qquad\text{and}\qquad
        t_w \bmod s = \alpha + \beta < n
    \end{equation}
    together imply
    \[
        q'(t_w)
        = q\left(\left\lfloor \frac{\alpha + \beta}{s} + 2 \tau_i \right\rfloor\right)
        = q(2 \tau_i)
        = \min e_i
        = a
        = c.
    \]
    If \(a \neq c\), such that \(W\)~is a path, then we need to distinguish two subcases based on its direction with respect to the edges \(e_i\)~and~\(e_j\).
    Informally, we will argue that \(W\)~is hit by the \(q\)-part of~\(q'\) if \(e_i\) comes before~\(e_j\), and by the \(r\)-part if \(e_i\) comes after~\(e_j\).
    Let first \(i < j\).
    Then, \(W\)~is covered by~\(q'\) for \(w = 0 \in \langle\gamma\rangle\), as the conditions~\eqref{eq:wg-1} and \(e_j \setminus e_i = \set{c}\) again imply that
    \[
        q'(t_w)
        = q\left(\left\lfloor \frac{\alpha + \beta}{s} + \tau_i + \tau_j \right\rfloor\right)
        = q(\tau_i + \tau_j)
        = c.
    \]
    Otherwise, we have \(i > j\), and \(W\)~is covered by~\(q'\) for \(w = n - \alpha - \beta \in \langle\gamma\rangle\), since then
    \[
        t_w \geq s > n
        \qquad\text{and}\qquad
        t_w \bmod s = n
    \]
    also imply that
    \[
        q'(t_w)
        = r\left(\left\lfloor \frac{n}{s} + \tau_i + \tau_j \right\rfloor\right)
        = r(\tau_i + \tau_j)
        = \ast(e_j \setminus e_i)
        = c.
    \]

    So far, we argued that every lazy walk of length at most two is hit by~\(q'\) if its total waiting time is at least~\(n\).
    It follows that any lazy walk of length at least three that is not hit by~\(q'\) must wait less than~\(n\) time steps at its first three vertices.
    We next argue that such a walk is hit by~\(q'\) as well.
    To this end, let \(W = (a, b, c, d)\) be a lazy walk with waiting times \((\alpha, \beta, \gamma, 0)\) such that \(\alpha + \beta + \gamma < n\).
    Let further \(W\)~use the edges
    \[
        e_i = \set{a, b},
        \quad
        e_j = \set{b, c},
        \quad\text{and}\quad
        e_k = \set{c, d},
    \]
    for some \(i, j, k \in [m]\).
    Suppose towards a contradiction that \(W\)~is not hit by~\(q'\),
    implying that
    \begin{align*}
        \forall w \in \langle\beta\rangle \colon q'(\alpha + s \tau_i + w) &\neq b, \\
        \forall w \in \langle\gamma\rangle \colon q'(\alpha + s \tau_i + \beta + s \tau_j + w) &\neq c, ~ \text{and} \\
        q'(\alpha + s \tau_i + \beta + s \tau_j + \gamma + s \tau_k) &\neq d.
    \end{align*}
    Fix \(w = 0\).
    Then, since \(\alpha + \beta + \gamma < n\) and \(s > n\), we have that
        \begin{align*}
            q'(\alpha + s \tau_i + w) &= q(\tau_i) \neq b, \\
            q'(\alpha + s \tau_i + \beta + s \tau_j + w) &= q(\tau_i + \tau_j) \neq c, ~ \text{and} \\
            q'(\alpha + s \tau_i + \beta + s \tau_j + \gamma + s \tau_k) &= q(\tau_i + \tau_j + \tau_k) \neq d.
    \end{align*}
    From here, a contradiction is reached using the same sequence of arguments as in the penultimate paragraph of the proof of \Cref{thm:baseline-general}.

    Let now \(W\)~be an unhit lazy walk with total delay~\(t_W\), and let \(\ell \defas \argmax_{i \in [m]} \tau_i\) be the index of the unique edge with maximum delay.
    We just established that \(W\)~traverses at most two edges and has total waiting time at most \(n - 1\), so we know that \(t_W \leq t^*\) for \(t^* \defas 2 s \tau_\ell + n - 1\).
    Equality can only occur if \(W\)~traverses \(e_\ell = \set{a, b}\) twice.
    Then, it takes the form \(W = (a, b, a)\) with waiting times \((\alpha, \beta, \gamma)\) such that \(\alpha + \beta + \gamma = n - 1\).
    If \(a > b\), then \(W\)~is hit after the first traversal of~\(e_\ell\) at time \(t \defas \alpha + s \tau_\ell > n\), since \(t \bmod s = \alpha \neq n\) implies
    \[
        q'(t) = q(\lfloor t/s \rfloor) = q(\tau_\ell) = \min e_\ell = b.
    \]
    If \(a < b\), then \(W\)~is hit following the second traversal of~\(e_\ell\) at time \(t \defas \alpha + \beta + 2 s \tau_\ell > n\), as again \(t \bmod s = \alpha + \beta \neq n\) and therefore
    \[
        q'(t) = q(\lfloor t/s \rfloor) = q(2 \tau_\ell) = \min e_\ell = a.
    \]
    Both cases contradict \(W\)~being unhit, and so it must be that \(t_W < t^*\).
    On the other hand, we have \(q'(t^*) = q(2 \tau_\ell) = \min e_\ell\), and therefore \((\tau', q')\) is winning as \(t_W < \max\supp(q')\).

    Since \(\supp(r) \subseteq \supp(q)\) by construction and
    \[
        t
        \leq s \lceil t / s \rceil
        \leq s (\lfloor t / s \rfloor + 1)
        \leq 2 s \lfloor t / s \rfloor
    \]
    holds for \(t > n\), the duration of the strategy is further bounded from above by
    \[
        \max\supp(q')
        \leq \max\set{n, 2 s \max\supp(q), 2 s \max\supp(r)}
        = 2 s \max\supp(q)
        \in \calO\bigl(n m^3\bigr).
        \qedhere
    \]
\end{proof}

\section{Starting late}%
\label{sec:waiting-freestart}

We next consider the most general setting where the evader, in addition to performing a lazy walk, has the power to start an unknown number of time steps before the pursuer begins their search.
This can also be viewed as the evader ``starting inside an edge,'' as we may consider an edge with delay~\(d\) to admit \(d-1\) intermediate positions that are occupied in sequence while the edge is traversed (see \Cref{fig:models}).
We show that the pursuer's ability to determine travel times still allows them to catch the evader on any graph, albeit after $\mathcal{O}(8^n)$ time steps in the worst case.
The same strategy applies to the setting where the evader starts early but cannot wait, and we leave it as an open question whether these settings can be separated.
When the playing field is restricted to trees, we give in \Cref{sec:waiting-freestart-trees} a strategy with an improved search time in~\(\mathcal{O}(4^n)\).

\subsection{A strategy for general graphs}%
\label{sec:waiting-freestart-general}

We first show how the pursuer can win on general graphs.
As a basis for our recursive strategy, we use the single-edge strategy shown in \Cref{fig:models} on \cpageref{fig:models}.
As the figure serves as a proof without words, we defer the formal argument to the appendix.

\begin{restatable}{lemma}{waitingfreestartedge}\label{lem:waiting-freestart-edge}
    The graph \(G = (\set{1, 2}, \set{\set{1, 2}})\) is waiting-freestart-solvable in time~\(4\).
\end{restatable}

We next extend the edge strategy to any graph.

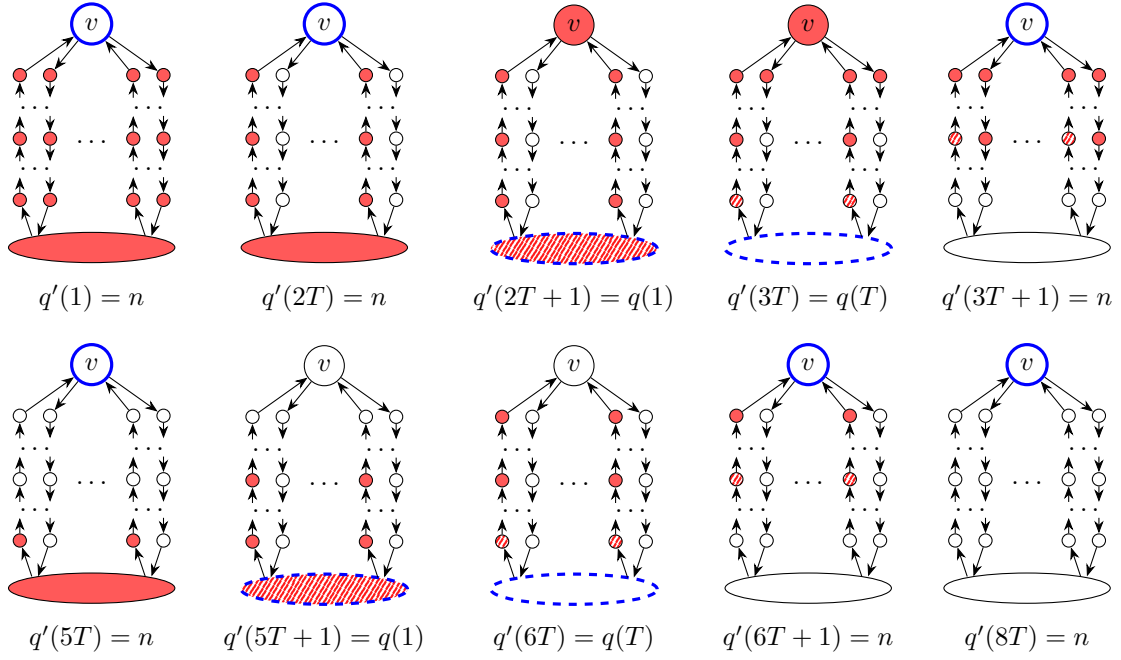
\begin{figure}
    \centering
    \newcommand{\completeGraphExtend}[8]{%
    \begin{tikzpicture}[baseline=(n.south)]
        \node[_named_vertex, #1] (n) at (0.00, 3*0.10) {$v$};
        \node (v1) at (-0.75, -3*0.85) {};
        \node (v2) at ( 0.75, -3*0.85) {};
        \node[
            draw,
            ellipse,
            minimum width=2.2cm,
            minimum height=0.4cm,
            #5%
        ] (oval) at (0, -3*0.86) { };

        \node[_vertex, _small, #2] (e1_l1)  at (-0.75+0.2, -3*0.1)  {};
        \node[_vertex, _small, #3] (e1_lT1) at (-0.75+0.2, -3*0.38) {};
        \node[_vertex, _small, #4] (e1_ln)  at (-0.75+0.2, -3*0.65) {};
        \node[_vertex, _small, #8] (e1_r1)  at (-0.75-0.2, -3*0.1)  {};
        \node[_vertex, _small, #7] (e1_rT1) at (-0.75-0.2, -3*0.38) {};
        \node[_vertex, _small, #6] (e1_rn)  at (-0.75-0.2, -3*0.65) {};

        \node[_vertex, _small, #2] (e2_l1)  at (0.75+0.2, -3*0.1)  {};
        \node[_vertex, _small, #3] (e2_lT1) at (0.75+0.2, -3*0.38) {};
        \node[_vertex, _small, #4] (e2_ln)  at (0.75+0.2, -3*0.65) {};
        \node[_vertex, _small, #8] (e2_r1)  at (0.75-0.2, -3*0.1)  {};
        \node[_vertex, _small, #7] (e2_rT1) at (0.75-0.2, -3*0.38) {};
        \node[_vertex, _small, #6] (e2_rn)  at (0.75-0.2, -3*0.65) {};

        \node at ( 0,     -3*0.4) {$\hdots$};
        \node at (-0.75, -3*0.25) {$\hdots$};
        \node at ( 0.75, -3*0.25) {$\hdots$};
        \node at (-0.75, -3*0.52) {$\hdots$};
        \node at ( 0.75, -3*0.52) {$\hdots$};

        \draw[_arc] (e1_ln) -- (v1);
        \draw[_arc] (v1) -- (e1_rn);
        \draw[_arc] (n) -- (e1_l1);
        \draw[_arc] (e1_r1) -- (n);
        \draw[_arc] {(e1_ln)+(0, 3*0.1)} -- (e1_ln);
        \draw[_arc] (e1_rn) -- (-0.75-0.2, -3*0.65 + 3*0.1);
        \draw[_arc] {(e1_lT1)+(0, 3*0.1)} -- (e1_lT1);
        \draw[_arc] (e1_rT1) -- (-0.75-0.2, -3*0.38 + 3*0.1);
        \draw[_arc] (e1_lT1) -- (-0.75+0.2, -3*0.38 - 3*0.1);
        \draw[_arc] (-0.75-0.2, -3*0.38 - 3*0.1) -- (e1_rT1);
        \draw[_arc] (e1_l1) -- (-0.75+0.2, -3*0.1 - 3*0.1);
        \draw[_arc] (-0.75-0.2, -3*0.1 - 3*0.1) -- (e1_r1);

        \draw[_arc] (e2_ln) -- (v2);
        \draw[_arc] (v2) -- (e2_rn);
        \draw[_arc] (n) -- (e2_l1);
        \draw[_arc] (e2_r1) -- (n);
        \draw[_arc] {(e2_ln)+(0, 3*0.1)} -- (e2_ln);
        \draw[_arc] (e2_rn) -- (0.75-0.2, -3*0.65 + 3*0.1);
        \draw[_arc] {(e2_lT1)+(0, 3*0.1)} -- (e2_lT1);
        \draw[_arc] (e2_rT1) -- (0.75-0.2, -3*0.38 + 3*0.1);
        \draw[_arc] (e2_lT1) -- (0.75+0.2, -3*0.38 - 3*0.1);
        \draw[_arc] (0.75-0.2, -3*0.38 - 3*0.1) -- (e2_rT1);
        \draw[_arc] (e2_l1) -- (0.75+0.2, -3*0.1 - 3*0.1);
        \draw[_arc] (0.75-0.2, -3*0.1 - 3*0.1) -- (e2_r1);
    \end{tikzpicture}%
}

\begin{tikzpicture}[
    every node/.style={font=\small},
    _small/.style={inner sep=1.75pt},
    b/.style={_burning},                
    p/.style={_possible},               
    h/.style={_hit},                    
    P/.style={_region_hit, _possible},  
    H/.style={_region_hit},             
]
    \matrix[
        matrix of nodes,
        column sep=1em,
    ] {
        \completeGraphExtend{h}{b}{b}{b}{b}{b}{b}{b} &
        \completeGraphExtend{h}{ }{ }{ }{b}{b}{b}{b} &
        \completeGraphExtend{b}{ }{ }{ }{P}{b}{b}{b} &
        \completeGraphExtend{b}{b}{ }{ }{H}{p}{b}{b} &
        \completeGraphExtend{h}{b}{b}{ }{ }{ }{p}{b} \\
        \node {$q'(1)=n$};                           &
        \node {$q'(2T)=n$};                          &
        \node {$q'(2T+1)=q(1)$};                     &
        \node {$q'(3T)=q(T)$};                       &
        \node {$q'(3T+1)=n$};                        \\[1ex]
        \completeGraphExtend{h}{ }{ }{ }{b}{b}{ }{ } &
        \completeGraphExtend{ }{ }{ }{ }{P}{b}{b}{ } &
        \completeGraphExtend{ }{ }{ }{ }{H}{p}{b}{b} &
        \completeGraphExtend{h}{ }{ }{ }{ }{ }{p}{b} &
        \completeGraphExtend{h}{ }{ }{ }{ }{ }{ }{ } \\
        \node {$q'(5T)=n$};                          &
        \node {$q'(5T+1)=q(1)$};                     &
        \node {$q'(6T)=q(T)$};                       &
        \node {$q'(6T+1)=n$};                        &
        \node {$q'(8T)=n$};                          \\
    };
\end{tikzpicture}
    \caption{%
        The winning strategy \((\tau', q')\) constructed in the proof of \Cref{thm:waiting-freestart-general} for a complete graph~\(K_{n-1}\) (oval) that is extended to the~\(K_n\) by adding the universal vertex~\(v\).
        All edges incident to~\(v\) are assigned a delay of~\(2T\), where \(T\)~is the duration of the sub-strategy \((\tau, q)\) solving the~\(K_{n-1}\).
        The first, \(T\)-th, and last vertices of the two directed paths representing an edge are shown.
        In addition to the keys given in \Cref{fig:models}, a dashed thick blue outline denotes that a query is made somewhere in the marked region, while striped red filling indicates that whether the evader can be present depends on the sub-strategy.
    }
    \label{fig:waiting-freestart-general}
\end{figure}

\begin{theorem}\label{thm:waiting-freestart-general}
    Every graph is waiting-freestart-solvable in time~\(\mathcal{O}(8^n)\).
\end{theorem}

\begin{proof}
    We construct a winning strategy with duration \(\lceil 8^n / 16 \rceil\) for the pursuer on the complete graph~\(K_n = (V, E)\).
    Since winning strategies remain winning if restricted to a subgraph, this implies the statement.

    For \(n = 1\), the pursuer wins by querying the only vertex at time \(1 = \lceil 8^n / 16 \rceil\).
    For \(n = 2\), we can employ the strategy of \Cref{lem:waiting-freestart-edge}, which has a duration of \(4 = 8^n / 16\) as claimed.

    For \(n \geq 3\), we may assume the induction hypothesis that there is a winning strategy \((\tau, q)\) for the pursuer on the \(K_{n-1} = (V, E)\) of duration \(T \defas \max\supp(q) = 8^{n - 1} / 16\).
    We extend it to a strategy \((\tau', q')\) for the~\(K_n = (V', E')\) as follows.
    Assume that vertices are labeled as \(V = [n - 1]\) and \(V' = [n]\), and let \(v \defas n\) be the unique vertex with \(v \in V' \setminus V\), whose incident edges are \(\delta(v) = E' \setminus E\).
    Let the new delays \(\tau' \colon E' \to \Zp\) be given by
    \[
        \tau'(e) \defas \begin{cases}
            \tau(e), & e \in E, \\
            2T, & e \in \delta(v),
        \end{cases}
    \]
    and define the query function \(q' \colon [8T] \to V'\) such that
    \[
        q'(t) \defas \begin{cases}
            q\bigl((t - 1) \bmod T + 1\bigr), & \lfloor (t - 1) / T \rfloor \in \set{2, 5}, \\
            v, & \lfloor (t - 1) / T \rfloor \in \set{0, 1, 3, 4, 6, 7}.
        \end{cases}
    \]
    This strategy is illustrated in \Cref{fig:waiting-freestart-general}.
    Less formally, \(q'\)~mimics~\(q\) on the subgraph \((V, E)\) in each of the length-\(T\) time horizons \([3T] \setminus [2T]\) and \([6T] \setminus [5T]\), and it queries the new vertex~\(v\) during the length-\(2T\) time horizons \([2T]\), \([5T] \setminus [3T]\), and \([8T] \setminus [6T]\).
    Note that \(q'\)~has a duration of \(8T = 8^n / 16\), so it only remains to show that \((\tau', q')\) is winning.

    Suppose towards a contradiction that the evader has a winning strategy starting at time \(t_{\mathrm{e}} \leq 0\), and let \(t_0\) with \(t_{\mathrm{e}} \leq t_0 \leq 0\) be the latest non-positive time at which they are located at a vertex.
    Denote this vertex by~\(v_0\).
    We first observe that the largest delay of any edge is~\(2T\).
    Otherwise, it must be that \(\tau(e) > 2T > \max\supp(q)\) for some edge \(e \in E\), but then \((\tau, q)\) could not be winning on the \(K_{n - 1}\) as the evader could enter~\(e\) at time~\(0\) to evade all queries.
    It follows that \(1 - t_0 \in [2T]\), as \(t_0 \leq -2T\) would contradict the choice of~\(t_0\).

    First consider the case that \(v_0 = v\).
    As the pursuer queries \(q'(1) = v\), the evader must leave~\(v\) along some edge \(\set{v, w} \in \delta(v)\) at a time~\(t_1\) with \(1 - 2T \leq t_0 \leq t_1 \leq 0\).
    The evader thus arrives at~\(w\) at time \(t_2 \defas t_1 + 2T \in [2T]\).
    If the evader remains on the subgraph \((V, E)\) at all times in \(\set{t_2, \ldots, 3T}\), then, as \(t_2 \leq 2T\), they are caught by the winning strategy that is executed on this subgraph in the time period \([3T] \setminus [2T]\).
    It follows that the evader must enter some edge \(e' \in \delta(v)\) in the direction of~\(v\) at a time \(t_3 \in [3T - 1]\), and arrives at~\(v\) at time \(t_4 \defas t_3 + 2T \in [5T - 1] \setminus [2T]\).
    Since \(q'(t) = v\) for all \(t \in [5T] \setminus [3T]\), we know that \(t_4 \in [3T] \setminus [2T]\), and further that the evader must leave~\(v\) along some edge \(\set{v, w'} \in \delta(v)\) at a time \(t_5 \in [3T] \setminus [2T]\).
    They then arrive at~\(w'\) at time \(t_6 \defas t_5 + 2T \in [5T] \setminus [4T]\).
    By an analogous argument to before, the evader cannot remain on the subgraph \((V, E)\) at all times in \(\set{t_6, \ldots, 6T}\), as a winning strategy is executed on this subgraph in the time period \([6T] \setminus [5T]\).
    Instead, they must enter some edge \(e'' \in \delta(v)\) in the direction of~\(v\) at a time \(t_7 \in [6T - 1] \setminus [4T]\), and arrive at~\(v\) at time \(t_8 \defas t_7 + 2T \in [8T - 1] \setminus [6T]\).
    But here they are immediately caught by a query, as \(q'(t) = v\) for all \(t \in [8T] \setminus [6T]\).

    It remains to consider the case where the evasive walk starts at a vertex \(v_0 \in V\).
    If the evader enters an edge \(e \in \delta(v)\) at a time~\(t\) with \(1 - 2T \leq t_0 \leq t \leq 0\), then they would reach vertex~\(v\) at time \(t + 2T \in [2T]\) and be hit by a query on arrival.
    On the other hand, if the evader remains on the subgraph \((V, E)\) at all times~\(t\) with \(t_0 \leq t \leq 3T\), then they would be caught by the copy of~\(q\) that is executed in the time period \([3T] \setminus [2T]\).
    Therefore, the evader must enter an edge \(e \in \delta(v)\) at some time \(t_1 \in [3T - 1]\).
    If they do so at a time \(t_1 \in [3T] \setminus [T]\), then they would arrive at vertex~\(v\) at time \(t_1 + 2T \in [5T] \setminus [3T]\), again being caught immediately.
    Hence, it must be that \(t_1 \in [T]\), so that the evader arrives at~\(v\) at some time \(t_2 \defas t_1 + 2T \in [3T] \setminus [2T]\).
    Waiting there until time \(3T + 1\) would again result in capture, so the evader must also enter some edge \(\set{v, w} \in \delta(v)\) at a time \(t_3 \in [3T] \setminus [2T]\), and arrive at vertex~\(w\) at time \(t_3 + 2T \in [5T] \setminus [4T]\).
    Remaining on the subgraph \((V, E)\) until time~\(6T\) would result in capture by the winning strategy executed in the time period \([6T] \setminus [5T]\), so the evader must enter an edge \(e' \in \delta(v)\) in the direction of~\(v\) at some time \(t_4 \in [6T - 1] \setminus [4T]\).
    It follows that the evader arrives at~\(v\) at time \(t_5 \defas t_4 + 2T \in [8T - 1] \setminus [6T]\), and is dispatched by a query immediately.

    As no combination of starting vertex~\(v_0\) and start time~\(t_{\mathrm{e}}\) allows the evader to elude all queries, it follows that \((\tau', q')\) is a winning strategy on \(K_n = (V', E')\).
\end{proof}

\subsection{A faster strategy for trees}%
\label{sec:waiting-freestart-trees}

When the graph being searched has no cycle, it is possible to improve the search time to~\(\mathcal{O}(4^n)\).
The pursuer's strategy is again built up inductively, starting from a single vertex and adding one leaf at a time.
The induction step is captured by the following lemma.

\begin{figure}
    \centering
    \newcommand{\treeExtend}[9]{%
    \begin{tikzpicture}[baseline=(v.south)]
        \node[_named_vertex,   #1] (v)     at ( 0.00,    0.00) {$v$};
        \node[_vertex, _small, #2] (e_l1)  at (+0.60, -3*0.15) {};
        \node[_vertex, _small, #3] (e_l2)  at (+0.60, -3*0.30) {};
        \node[_vertex, _small, #4] (e_ln)  at (+0.60, -3*0.70) {};
        \node[_vertex, _small, #8] (e_r1)  at (-0.60, -3*0.15) {};
        \node[_vertex, _small, #7] (e_r2)  at (-0.60, -3*0.30) {};
        \node[_vertex, _small, #6] (e_rn)  at (-0.60, -3*0.70) {};
        \node[_named_vertex,   #5] (w)     at ( 0.00, -3*0.85) {$w$};
        \node[                   ] (dots1) at (+0.60, -3*0.55) {$\vdots$};
        \node[                   ] (dots2) at (-0.60, -3*0.55) {$\vdots$};
        \node[
            regular polygon,
            regular polygon sides=3,
            draw,
            minimum size=1.5cm,
            below=0pt of w,
            #9%
        ] (A) {};

        \draw[_arc] (e_ln)           -- (w);
        \draw[_arc] (w)              -- (e_rn);
        \draw[_arc] (e_l1)           -- (e_l2);
        \draw[_arc] (v)              -- (e_l1);
        \draw[_arc] (e_r1)           -- (v);
        \draw[_arc] (e_r2)           -- (e_r1);
        \draw[_arc] (e_rn)           -- (-0.60, -3*+0.6);
        \draw[_arc] (-0.60, -3*+0.4) -- (e_r2);
        \draw[_arc] (e_l2)           -- (+0.60, -3*+0.4);
        \draw[_arc] (+0.60, -3*+0.6) -- (e_ln);
    \end{tikzpicture}%
}

\begin{tikzpicture}[
    every node/.style={font=\small},
    _small/.style={inner sep=1.75pt},
    b/.style={_burning},                
    p/.style={_possible},               
    h/.style={_hit},                    
    P/.style={_region_hit, _possible},  
    H/.style={_region_hit},             
]
    \matrix[
        matrix of nodes,
        column sep=1em,
    ] {
        \treeExtend{b}{b}{b}{b}{h}{b}{b}{b}{b} &
        \treeExtend{b}{b}{b}{b}{h}{ }{ }{b}{b} &
        \treeExtend{h}{b}{b}{b}{b}{ }{ }{ }{b} &
        \treeExtend{ }{ }{b}{b}{h}{b}{ }{ }{b} &
        \treeExtend{ }{ }{ }{ }{h}{ }{ }{b}{b} \\
        \node {$q'(1)=w$};                     &
        \node {$q'(T)=w$};                     &
        \node {$q'(T+1)=v$};                   &
        \node {$q'(T+2)=w$};                   &
        \node {$q'(2T+1)=w$};                  \\[1ex]
        \treeExtend{h}{ }{ }{ }{b}{ }{ }{ }{b} &
        \treeExtend{ }{ }{ }{ }{P}{b}{ }{ }{P} &
        \treeExtend{ }{ }{ }{ }{H}{p}{p}{b}{H} &
        \treeExtend{h}{ }{ }{ }{ }{ }{p}{p}{ } &
        \treeExtend{h}{ }{ }{ }{ }{ }{ }{ }{ } \\
        \node {$q'(2T+2)=v$};                  &
        \node {$q'(2T+3)=q(1)$};               &
        \node {$q'(3T+2)=q(T)$};               &
        \node {$q'(3T+3)=v$};                  &
        \node {$q'(4T+2)=v$};                  \\
    };
\end{tikzpicture}
    \caption{%
        The winning strategy \((\tau', q')\) given in \Cref{lem:waiting-freestart-leaf} for a tree (triangle) that is extended at~\(w\) with a leaf vertex~\(v\).
        The edge \(\set{v, w}\) has a delay of \(\tau'(\set{v, w}) = T + 1\).
        See \Cref{fig:models,fig:waiting-freestart-general} for keys.
    }
    \label{fig:waiting-freestart-leaf}
\end{figure}
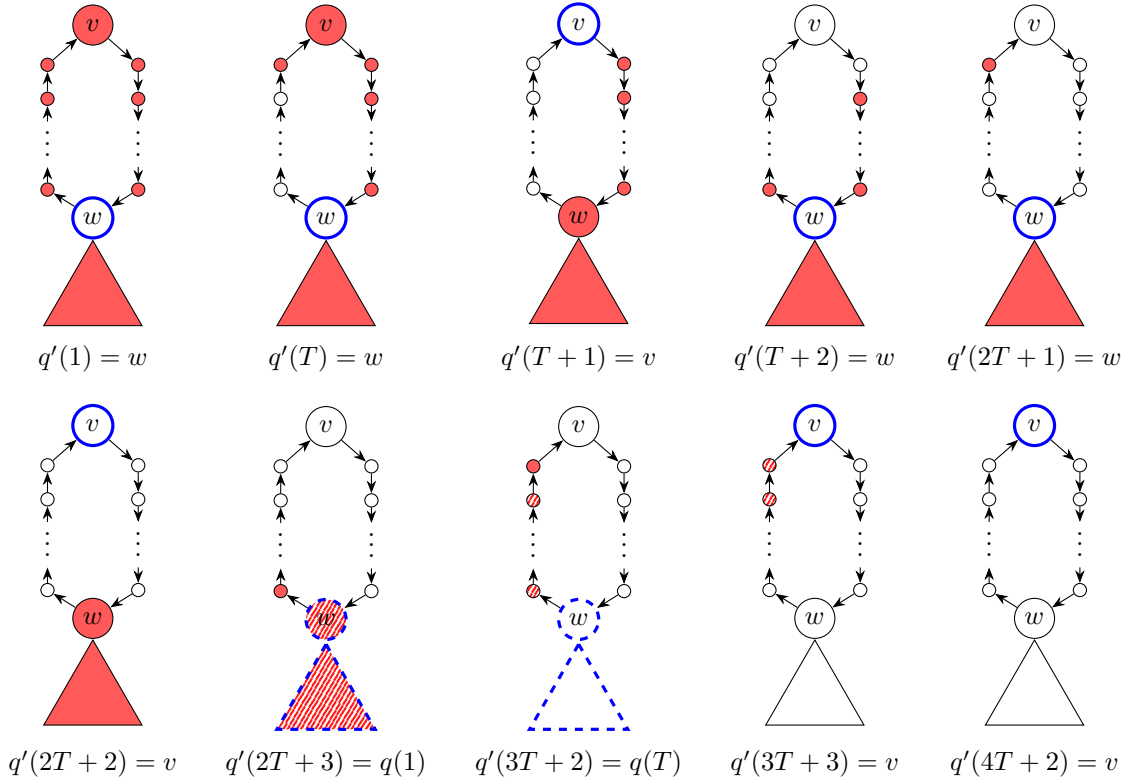

\begin{lemma}\label{lem:waiting-freestart-leaf}
    Let \(G = (V, E)\) be a waiting-freestart-solvable graph and let \((\tau, q)\) be a winning strategy with duration~\(T\).
    Consider the graph \(G' = (V \cup \set{v}, E \cup \smash{\bset{\set{v, w}}})\) obtained by attaching a pendant vertex~\(v\) to some vertex \(w \in V\).
    Then, \(G'\)~is waiting-freestart-solvable in time \(4T + 2\).
\end{lemma}

\begin{proof}
    Define \(V' \defas V \cup \set{v}\) and \(E' \defas E \cup \set{e}\) with \(e \defas \set{v, w}\).
    We construct a winning strategy \((\tau', q')\) for \(G' = (V', E')\) as follows.
    Let the new delays \(\tau' \colon E' \to \Zp\) be given by
    \[
        \tau'(e) \defas \begin{cases}
            \tau(e), & e \in E, \\
            T + 1, & e = \set{v, w},
        \end{cases}
    \]
    and define the query function \(q' \colon [4T + 2] \rightharpoonup V'\) such that
    \[
        q'(t) \defas \begin{cases}
            w, & t \in [2T + 1] \setminus \set{T + 1}, \\
            v, & t \in [4T + 2] \setminus [3T + 2] \lor t \in \set{T + 1, 2T + 2}, \\
            q\bigl(t - (2T + 2)\bigr), & t \in [3T + 2] \setminus [2T + 2].
        \end{cases}
    \]
    The strategy is visualized in \Cref{fig:waiting-freestart-leaf}.
    Clearly, \(\max\supp(q') = 4T + 2\).
    The remainder of the proof follows the same line of arguments as the induction step in the proof of \Cref{thm:waiting-freestart-general}.

    Suppose towards a contradiction that the evader has a winning strategy starting at time \(t_{\mathrm{e}} \leq 0\), and let \(t_0\) with \(t_{\mathrm{e}} \leq t_0 \leq 0\) be the latest non-positive time at which they are located at a vertex.
    Denote this vertex by~\(v_0\).
    We first observe that the largest delay of any edge is \(T + 1\).
    Otherwise, it must be that \(\tau(e) > T + 1 > \max\supp(q)\) for some edge \(e \in E\), but then \((\tau, q)\) could not be winning on~\(G\) as the evader could enter~\(e\) at time~\(0\) to evade all queries.
    It follows that \(-t_0 \in \langle{T}\rangle\), as \(t_0 \leq -(T + 1)\) would contradict the choice of~\(t_0\).

    First consider the case that \(v_0 = v\).
    As the pursuer queries \(q'(T + 1) = v\), the evader must leave~\(v\) along~\(e\) at a time~\(t_1\) with \(-T \leq t_0 \leq t_1 \leq T\).
    The evader thus arrives at~\(w\) at time \(t_2 \defas t_1 + (T + 1) \in [2T + 1]\).
    As \(q'(t) = w\) for all \(t \in [2T + 1] \setminus \set{T + 1}\), it must be that \(t_2 = T + 1\).
    Since \(q'(t_2 + 1) = q'(T + 2) = w\), the evader cannot stay at~\(w\) for another time step.
    If they entered~\(e\) again at time~\(t_2\), they would arrive back at~\(v\) at time \(t_2 + (T + 1) = 2T + 2\) and be caught by the one-off query \(q'(2T + 2) = v\).
    The evader thus enters an edge \(e' = \set{w, w'}\) towards a vertex \(w' \in V \setminus \set{w}\) at time~\(t_2\).
    If the evader then stays within the subgraph~\(G\) until time \(3T + 2\), they would be caught by the copy of the winning strategy executed in the time range \([3T + 2] \setminus [2T + 2]\).
    It follows that the evader enters~\(e\) towards~\(v\) at a time~\(t_3\) with \(t_2 \leq t_3 \leq 3T + 1\).
    If \(t_3 \leq 2T + 1\), then the evader would be located at~\(w\) at time \(t_3 \in [2T + 1] \setminus [t_2]\), and be hit by the query \(q'(t_3) = w\).
    Thus, \(t_3 \in [3T + 1] \setminus [2T + 1]\), and so the evader arrives at~\(v\) at time \(t_4 \defas t_3 + (T + 1) \in [4T + 2] \setminus [3T + 2]\).
    But this is not possible in an evasive walk, as \(q'(t_4) = v\).

    It remains to consider the case where the evasive walk starts at a vertex \(v_0 \in V\).
    If the evader stays within the subgraph~\(G\) until time \(3T + 2\) then they are caught by the sub-strategy, so they must be located at~\(w\) and enter~\(e\) at a time~\(t_1\) with \(-T \leq t_0 \leq t_1 \leq 3T + 1\).
    We can rule out \(t_1 \in [2T + 1] \setminus \set{T + 1}\) as otherwise the evader is hit by the query \(q'(t_1) = w\).
    If \(t_1 = T + 1\), then the evader arrives at~\(v\) at time \(t_1 + (T + 1) = 2T + 2\) and is hit by the one-off query at that time.
    If \(t_1 \in [3T + 1] \setminus [2T + 1]\), then the evader arrives at~\(v\) at time \(t_2 \defas t_1 + (T + 1) \in [4T + 2] \setminus [3T + 2]\), where they are again caught by the query \(q'(t_2) = v\).
    It follows that \(-T \leq t_1 \leq 0\), so that the evader reaches~\(v\) at time \(t_2' \defas t_1 + (T + 1) \in [T + 1]\).
    As \(q'(T + 1) = v\), it must be that \(t_2' \in [T]\), and further that the evader leaves~\(v\) again towards~\(w\) at some time \(t_3 \geq t_2'\) with \(t_3 \leq T\).
    But then, they arrive at~\(w\) at time \(t_4 = t_3 + (T + 1) \in [2T + 1] \setminus [T + 1]\), where \(q'(t_4) = w\).
\end{proof}

Applied to subgraphs that are trees, \Cref{lem:waiting-freestart-leaf} implies the following search time.

\begin{theorem}\label{thm:waiting-freestart-trees}
    Trees are waiting-freestart-solvable in time~\(\calO(4^n)\).
\end{theorem}

\begin{proof}
    We show by induction on~\(n\) that every tree on \(n\)~vertices admits a winning strategy with duration \(T(n) = (5 \cdot 4^{n-1} - 2)/3\).
    If \(n = 1\), a single query suffices, so \(T(n) = 1 = (5 \cdot 1 - 2)/3\).
    For \(n \geq 2\), applying \Cref{lem:waiting-freestart-leaf} to an arbitrary leaf yields a winning strategy of duration
    \[
        T(n)
        = 4T(n-1) + 2
        = \frac{4}{3} (5 \cdot 4^{n-2} - 2) + 2
        = \frac{1}{3} (5 \cdot 4^{n-1} - 2). \qedhere
    \]
\end{proof}

\section{Discussion}

We studied a pursuit-evasion game on undirected graphs in which a single pursuer has the power to control the travel times on each edge and then query an arbitrary vertex per time step.
We gave explicit winning strategies for the pursuer on all graphs, even if the evader can wait at vertices or start their walk at a time in the past unknown to the pursuer.
If the evader must start at time zero, then our strategies require only a polynomial number of time steps in the size of the graph, while in the free-start setting, we only know an exponential-duration strategy, even for trees.
Both of the free-start results assume that the evader can wait, and it is an interesting question whether a polynomial-duration solution exists if they cannot.
A potential tool for exploiting forced moves is parity: assign delays such that the vertices admit a bipartition with the property that walks ending in the same part have even total delay, whereas walks starting and ending in different parts have odd total delay.
Then, one can first execute a strategy assuming that the evader started in one part, and then repeat the same strategy for the other part.
We could derive full algorithms from this idea, but they were outclassed by \Cref{thm:waiting-freestart-general,thm:waiting-freestart-trees}.

A computer search that identified much shorter solutions for graphs of order at most four (see \Cref{sec:computations}) sparks hope that a polynomial-duration strategy might exist even in the general setting where the evader is most powerful: the complete graph~\(K_n\) with \(n\)~ranging from~\(2\) to~\(4\) is waiting-freestart-solvable in times \(4\), \(14\), and~\(32\), whereas the strategy of \Cref{thm:waiting-freestart-general} has respective durations of \(4\), \(32\), and~\(256\).
A complementary question worth pursuing is whether the provided polynomial search times are asymptotically optimal for the fixed-start settings.

\bibliographystyle{abbrvnat}
\setlength{\bibsep}{6pt plus 1pt minus 1pt}
\bibliography{manuscript}

\appendix
\crefalias{section}{appendix}

\section{Deferred proofs}%
\label{sec:deferred}

\bhshifting*

\begin{proof}
    Since \(A = \set{a_1, \ldots, a_m}\) is a \(B_h\)-set, we have for any \(c, d \in \Zn^m\) with \(\lVert c \rVert_1 = \lVert d \rVert_1 = h\) (implying \(\lVert c \rVert_\infty, \lVert d \rVert_\infty \leq h\)) that
    \[
        \sum_{i=1}^m c_i a_i = \sum_{i=1}^m d_i a_i \implies c = d.
    \]
    Let \(t \defas \min A\) and assume without loss of generality that \(t = a_m\).
    As
    \[
        \sum_{i=1}^m c_i (a_i - t) = \sum_{i=1}^m d_i (a_i - t)
        ~\Longleftrightarrow~
        \sum_{i=1}^m c_i a_i - ht = \sum_{i=1}^m d_i a_i - ht
        ~\Longleftrightarrow~
        \sum_{i=1}^m c_i a_i = \sum_{i=1}^m d_i a_i,
    \]
    also \(A' \defas \set{a - t \mid a \in A}\) is a \(B_h\)-set with \(0 \in A'\).
    Label \(A' = \set{a_1', \ldots, a_m'}\) such that \(a_i' = a_i - t\) for all \(i \in [m]\), in particular \(a_m' = 0\), and suppose towards a contradiction that \(A' \setminus \set{0} = \set{a_1', \ldots, a_{m-1}'}\) is not a \(B_{0,h}^h\)-set.
    Then, there are coefficients \(c \neq d \in \Zn^{m-1}\) with \(\lVert c \rVert_1, \lVert d \rVert_1 \leq h\) (implying \(\lVert c \rVert_\infty, \lVert d \rVert_\infty \leq h\)), such that
    \[
        \sum_{i=1}^{m-1} c_i a_i' = \sum_{i=1}^{m-1} d_i a_i'.
    \]
    Consider now coefficients \(c', d' \in \Zn^m\) with
    \[
        c_i' \defas \begin{cases}
            c_i, & i < m, \\
            h - \lVert c \rVert_1, & i = m,
        \end{cases}
        \qquad\text{and}\qquad
        d_i' \defas \begin{cases}
            d_i, & i < m, \\
            h - \lVert d \rVert_1, & i = m.
        \end{cases}
    \]
    By construction, we have that \(\lVert c' \rVert_1 = \lVert d' \rVert_1 = h\), and further
    \[
        \sum_{i=1}^m c_i' a_i'
        = \sum_{i=1}^{m-1} c_i' a_i' + c_m' a_m'
        = \sum_{i=1}^{m-1} c_i a_i'
        = \sum_{i=1}^{m-1} d_i a_i'
        = \sum_{i=1}^{m-1} d_i' a_i' + d_m' a_m'
        = \sum_{i=1}^m d_i' a_i'.
    \]
    Since \(A'\)~is a \(B_h\)-set, it follows that \(c' = d'\) and by extension that \(c = d\), a contradiction.
    It follows that \(A' \setminus \set{0}\) must be a \(B_{0,h}^h\)-set, and by definition also a \(B_{g,h}^r\)-set for any \(g, r \leq h\).
\end{proof}

\waitingfreestartedge*

\begin{proof}
    Let \(G = (V, E)\) as in the statement, and let \(e \in E\) be the only edge.
    Assign a delay of \(\tau(e) \defas 2\) and query according to \(q \colon [4] \to V\) with \(q(t) \defas t \bmod 2 + 1\).
    Suppose that the evader has a winning strategy starting at time \(t_{\mathrm{e}} \leq 0\), and let~\(t_0\) with \(t_{\mathrm{e}} \leq t_0 \leq 0\) be the latest non-positive time at which they are located at a vertex.
    Let further \(v_0 \in \set{1, 2}\) be this vertex.
    As the largest delay of any edge is \(\tau(e) = 2\), we have that \(t_0 \in \set{-1, 0}\).
    Suppose first that \(t_0 = -1\).
    This implies that the evader enters the edge at time~\(t_0\), as otherwise they would still be at a vertex at time~\(0\), contradicting the choice of~\(t_0\).
    If \(v_0 = 1\), then the evader arrives at vertex~\(2\) at time \(t_0 + \tau(e) = 1\), where they are hit by a query as \(q(1) = 2\).
    If \(v_0 = 2\), then the evader instead arrives at vertex~\(1\) at time~\(1\).
    If they wait there for one time step, then they are hit as \(q(2) = 1\), so they must enter the edge and return to vertex~\(2\) at time~\(3\).
    But here they are also caught, as \(q(3) = 2\).
    It follows that \(t_0 = 0\).
    Suppose first that \(v_0 = 1\).
    If the evader waits at~\(v_0\) until time~\(2\), then they are caught by the query \(q(2) = 1\).
    If they wait for one time step to depart at time~\(1\), then they arrive at vertex~\(2\) at time \(1 + \tau(e) = 3\), and are hit by \(q(3) = 2\).
    Lastly, if they enter the edge immediately, the evader arrives at vertex~\(2\) at time \(\tau(e) = 2\).
    Here, they cannot stay for another time step, as \(q(3) = 2\), but returning immediately plays into the query \(q(2 + \tau(e)) = q(4) = 1\).
    In the remaining case of \(t_0 = 0\) and \(v_0 = 2\), the evader has two options.
    If they wait for one time step, then they are caught by the query \(q(1) = 2\).
    If they enter the edge immediately, then they arrive at vertex~\(1\) at time~\(2\), but this has the same outcome due to \(q(2) = 1\).
    It follows that \((\tau, q)\) is a winning strategy for the pursuer with duration \(\max\supp(q) = 4\).
\end{proof}

\section{Short strategies for the most general setting}%
\label{sec:computations}

In \Cref{fig:k3_14,fig:k4_32} on \cpageref{fig:k3_14,fig:k4_32}, we give explicit winning strategies for the complete graphs \(K_3\) and~\(K_4\), with respective durations of \(14\) and~\(32\), for the general setting where the evader may both wait at a vertex and start at an unknown time in the past.
Recall that an explicit strategy with duration~\(4\) for the~\(K_2\) was given in \Cref{fig:models} on \cpageref{fig:models}.
We conjecture that these strategies are the shortest possible, but know this (from complete enumeration) only for a maximum delay of nine.
This contrasts with the recursive strategy of \Cref{thm:waiting-freestart-general}, which also takes \(8^2/16 = 4\) time steps on the~\(K_2\), but \(8^3/16 = 32\) time steps on the~\(K_3\) and \(8^4/16 = 256\) steps on the~\(K_4\).

\begin{figure}
    \centering
    \input{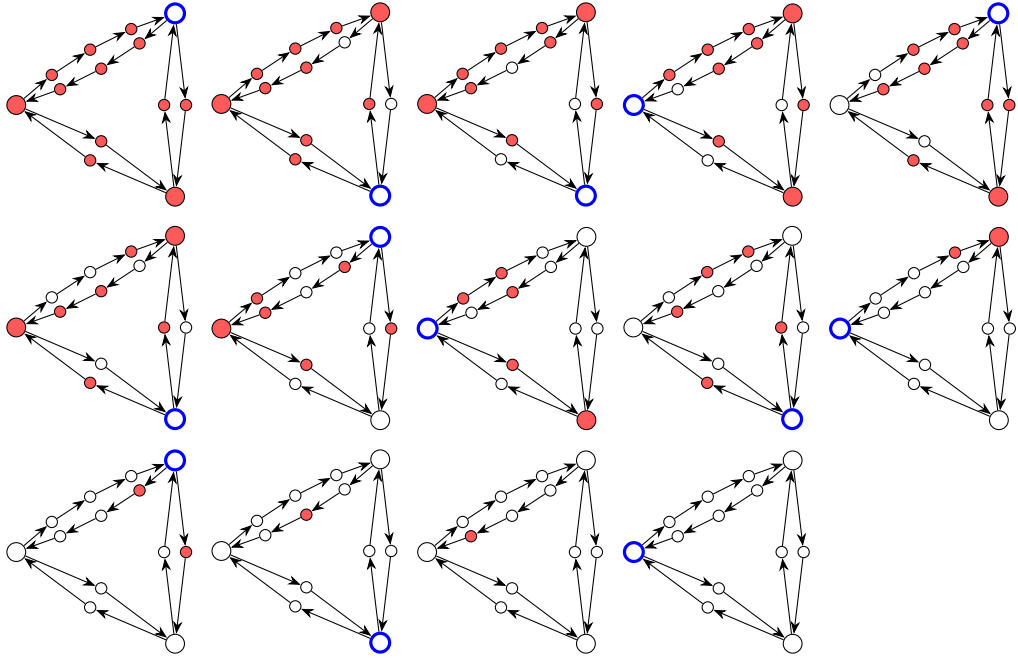}
    \caption{%
        A strategy against a waiting and free-starting evader on the~\(K_3\) of duration~\(14\), with edge delays of \(2\), \(2\), and \(4\).
        See \Cref{fig:models} for keys.
    }
    \label{fig:k3_14}
\end{figure}

\begin{figure}
    \tikzset{_arc/.style={-stealth}}
    \centering
    \input{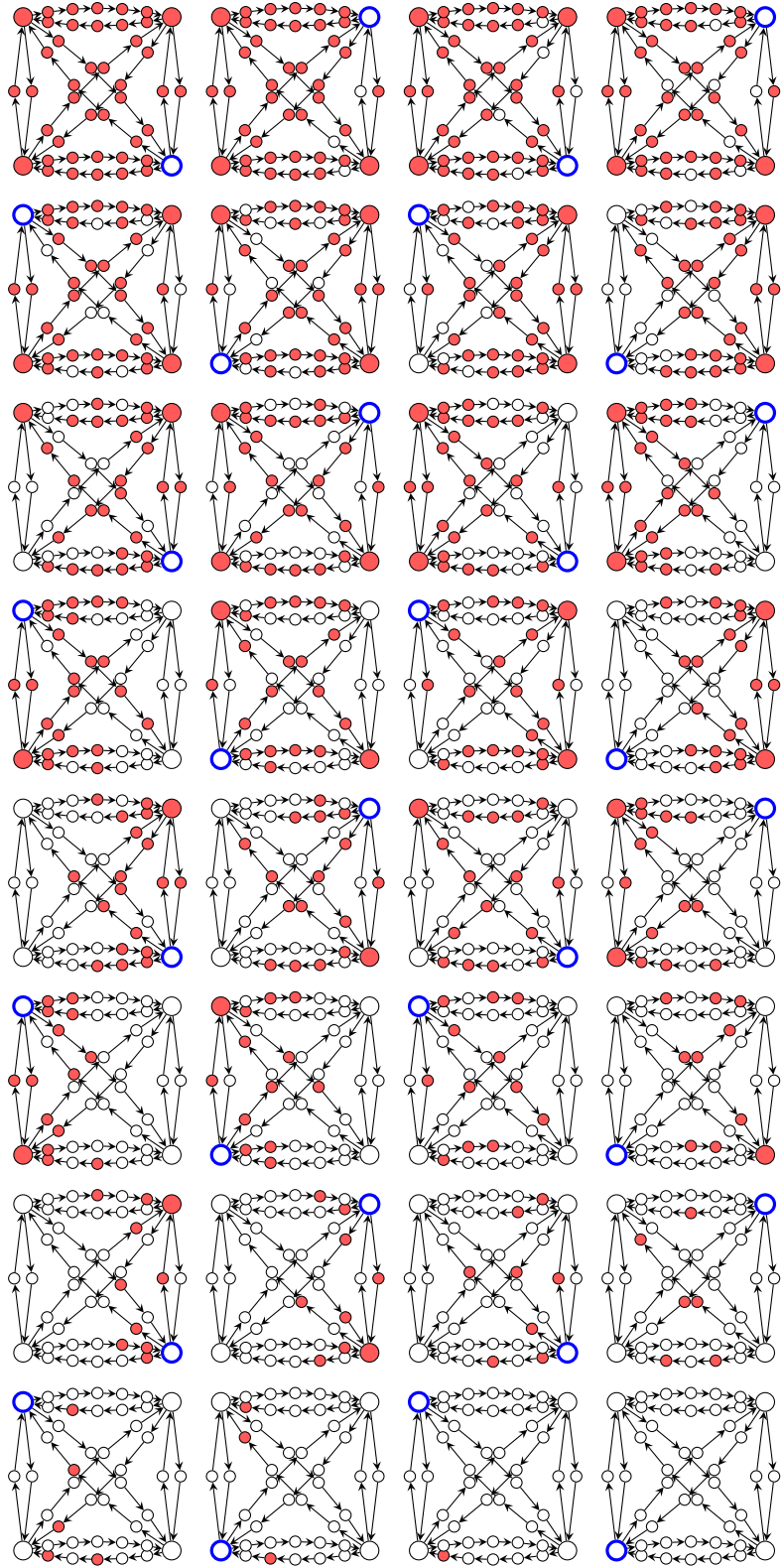}
    \caption{%
        A strategy against a waiting and free-starting evader on the~\(K_4\) of duration~\(32\), with edge delays of \(2\), \(5\), and \(6\), each occurring twice.
        See \Cref{fig:models} for keys.
    }
    \label{fig:k4_32}
\end{figure}

\end{document}